\documentclass[pdflatex,sn-mathphys-num]{sn-jnl}% Math and Physical Sciences Numbered Reference Style

\usepackage{graphicx}%
\usepackage{multirow}%
\usepackage{amsmath,amssymb,amsfonts}%
\usepackage{amsthm}%
\usepackage{mathrsfs}%
\usepackage[title]{appendix}%
\usepackage{xcolor}%
\usepackage{textcomp}%
\usepackage{manyfoot}%
\usepackage{booktabs}%
\usepackage{algorithm}%
\usepackage{algorithmicx}%
\usepackage{algpseudocode}%
\usepackage{listings}%

\usepackage[mathscr]{euscript}
\usepackage{mathtools}
\usepackage{enumerate}

\usepackage{verbatim} 
\usepackage[normalem]{ulem}
\usepackage{cancel}

\usepackage[numbers,sort&compress]{natbib}

\usepackage{csquotes}

\newcommand{\SC}{\Delta_{SC}}
\newcommand{\CS}{\Delta_{CS}}
\newcommand{\Vint}{\mathring{V}}

\newcommand{\bt}{\mathcal{BT}^*}
\newcommand{\fb}{T^{fb}}
\newcommand{\gfb}{T^{gfb}}
\newcommand{\mb}{T^{mb}}
\newcommand{\cat}{T^{cat}}

\theoremstyle{thmstyleone}
\newtheorem{theorem}{Theorem}
\newtheorem{proposition}{Proposition}
\newtheorem{corollary}{Corollary}
\newtheorem{lemma}{Lemma}

\theoremstyle{thmstyletwo}

\newtheorem{remark}{Remark}

\theoremstyle{thmstylethree}

\begin{document}

\title{Revealing the building blocks of tree balance: fundamental units of the Sackin and Colless Indices}

\affil[1]{Institute of Mathematics and Computer Science, University of Greifswald, Walther-Rathenau-Str. 47, 17487 Greifswald, Germany}

\author[1]{\fnm{Linda} \sur{Kn{\"u}ver}}

\author*[1]{\fnm{Mareike} \sur{Fischer}}\email{mareike.fischer@uni-greifswald.de,email@mareikefischer.de}

\abstract{(Im)balance indices can be used to quantify the (im)balance of trees by assigning numerical scores to them. An easy way to generate a new index is to construct a compound index, e.g., a linear combination of established indices. Two of the most prominent and widely used imbalance indices are the Sackin index and the Colless index. In this study, we show that these classic indices are themselves compound in nature: they can be decomposed into more elementary components that independently satisfy the defining properties of a tree (im)balance index. 
We further show that the difference Colless minus Sackin results in another imbalance index that is minimized (amongst others) by all Colless minimal trees. Conversely, the difference Sackin minus Colless forms a balance index. Finally, we compare the building blocks of which the Sackin and the Colless indices consist to these indices as well as to the stairs2 index, which is another index from the literature. Our results suggest that the elementary building blocks we identify are not only foundational to established indices but also valuable tools for analyzing disagreement among indices when comparing the balance of different trees. Along the way, we investigate the so-called echelon tree, which plays an important role for several (im)balance indices, and present the first non-recursive algorithm to construct it.}


\maketitle

\section{Introduction}\label{sec:Intro}
Quantifying the balance of phylogenetic trees is a central task in evolutionary biology and related fields, with applications ranging from model selection to the inference of diversification dynamics \cite{Fischer2023,Kersting2025}. Over the years, more than two dozen tree balance indices have been introduced \cite{Fischer2023}, and even some infinite families of such indices are known \cite{cleary2025}. These indices differ not only in construction but also in their statistical and biological behavior, such as their sensitivity to particular tree shapes or their ability to distinguish between different generative models of tree evolution \cite{Blum2005,Kersting2025}.

This variability has led to the continuing development of new indices, often with the aim of capturing specific aspects of tree balance more effectively. A common approach in the design of new indices is to combine established ones into compound indices, for instance via simple linear combinations \cite[Chapter 25]{Fischer2023}, \cite{kersting2020,Cardona2013,Matsen2006}. Among the most commonly used base indices for such compound construction are the Sackin index \cite{Sackin1972,Fischer2021}, which measures the total of all subtree sizes, and the Colless index \cite{Colless1982,coronado2020}, which measures subtree size imbalances across all internal nodes. Both have been used extensively in theoretical studies and practical applications \cite{Blum2005,Blum2006,Wicke2020,Fischer2023, Kayondo2021}, and they often serve as components in the construction of new indices \cite{Fischer2023,Cardona2013,kersting2020,Matsen2006}.

In this work, we take a step back and ask: How elementary are the Sackin and Colless indices themselves? Surprisingly, we find that both indices can be decomposed into even more basic components (in the sense of using less information of the underlying tree) that, on their own, fulfill the defining criteria of  (im)balance indices. In order to see this,  for each internal node $v$ in a binary rooted tree, we consider the sizes $n_{v_a}$ and $n_{v_b}$ of its two maximal pending subtrees (with $n_{v_a}\geq n_{v_b}$), i.e., of the two subtrees rooted at the children $a$ and $b$ of $v$. We show that the total sum $N_a$ of all $n_{v_a}$ values, $N_a=\sum\limits_{v\in \mathring{V}}n_{v_a}$,  defines an imbalance index, while the corresponding sum $N_b$ of all $n_{v_b}$ values, $N_b=\sum\limits_{v\in \mathring{V}}n_{v_b}$,  defines a balance index. Then, it can be easily shown that the Sackin index is the sum of $N_a$ and $N_b$, while the Colless index is the difference of these values. These decompositions reveal that both Sackin and Colless themselves are in fact mere compound indices built from these more elementary units, which themselves  are (im)balance indices.

Building on this decomposition of the Sackin and Colless indices, we even identify two more  indices, which are new to the literature and which turn out to be closely related to $N_a$ and $N_b$: the difference Colless-Sackin, $\Delta_{CS}$, forms an imbalance index that remains minimal (amongst others) for the same class of balanced trees that minimize Colless, while the difference Sackin-Colless, $\Delta_{SC}$, constitutes a balance index. Concerning all newly discovered indices, i.e., $N_a$, $N_b$, $\Delta_{CS}$ and $\Delta_{SC}$, we are able to fully characterize their extremal trees for all $n$, and in this characterization, the so-called echelon tree plays a decisive role. The echelon tree has so far mainly been known from the context of two other tree balance indices from the literature, namely the so-called stairs2 and $B_1$ balance indices \cite{Currie2024,fischer2026}. In our study, we also present some new insight into this tree and present, for instance, the first non-recursive way to generate and define it.

We compare all new indices mentioned above with the original Sackin and Colless indices, as well as the above mentioned stairs2 index. We argue that the newly identified building blocks -- the sums $N_a$ and $N_b$ -- can serve as insightful tools for investigating cases where different indices disagree on the relative balance of two trees, thereby offering new perspectives for index development and evaluation. 

\section{Preliminaries} \label{sec:Prelim}

\subsection{Definitions and notation}\label{prelim:def}
In this section, we provide the required preliminary knowledge needed to establish our results later on. We start with some basic definitions.\\

\noindent\textbf{Basic tree concepts.} 
The most important concept for the present manuscript is that of a \emph{rooted binary tree} (or simply \emph{tree} for short), which is a directed graph $T = (V(T),E(T))$, with vertex set $V(T)$ and edge set $E(T)$, containing precisely one vertex of in-degree zero, namely the \emph{root} (denoted by $\rho$), such that for every $v \in V(T)$ there exists a unique path from $\rho$ to $v$, and such that each vertex either has has out-degree 2 (such vertices will be referred to as \emph{inner vertices}) or 0 (such vertices will be referred to as \emph{leaves}). For every $n \in \mathbb{N}_{\geq 1}$, we denote by $\bt_n$ the set of (isomorphism classes of) rooted binary trees with $n$ leaves. 

In the following, we use $V^1(T) \subseteq V(T)$ to refer to the leaf set of $T$ (i.e., $V^1(T) = \{v \in V(T): \text{out-degree}(v)=0\}$), and we use $\mathring{V}(T)$ to denote the set of inner vertices of $T$ (i.e., $\mathring{V}(T) = V(T) \setminus V^1(T)$). Moreover, we use $n$ to denote the number of leaves of $T$, i.e., $n = \vert V^1(T) \vert$. Note that $\rho \in \mathring{V}(T)$ if $n \geq 2$. If $n=1$, $T$ consists of only one vertex, which is at the same time the root and its only leaf.

Now, as we only consider binary trees, recall that we can decompose each one of them into its two maximal pending subtrees, i.e., $T=(T_a, T_b)$, where $T_a$ and $T_b$ denote the two subtrees adjacent to the root of $T$. This \emph{standard decomposition} can also be considered for pending subtrees of $T$, i.e., for subtrees $T_v$ rooted at an inner vertex $v$ of $T$ and containing all \emph{descendants} of $v$. The number of leaves in this tree will be denoted by $n_v$. Note that a descendant of $v$ is a vertex $w$ of $T$ for which $v$ lies on the unique path from the root $\rho$ of $T$ to $w$. Likewise, if $w$ is a descendant of $v$, $v$ is an \emph{ancestor} of $w$. If the distance between an ancestor $v$ and its descendant $w$ is precisely one, we also call $v$ the \emph{parent} of $w$ and $w$ the \emph{child} of $v$, 
and two leaves $x$ and $y$ with the same parent are said to form a \emph{cherry} $[x,y]$.

Recall that whenever you have a rooted tree $T$, you can define the \emph{depth} $\delta_v$ of a vertex $v$ of $T$ as the length of the unique path from the root $\rho$ of $T$ to $v$. The \emph{height} $h(T)$ of $T$ is then simply the maximum of all leaf depths, i.e., $h(T)=\max_{x \in V^1(T)} \delta_x$. Given a subset $X'$ of leaves of $T$, the \emph{lowest common ancestor} or \emph{lca} of $X'$ is the unique ancestor of all leaves in $X'$ with the largest depth.\\

\noindent\textbf{ Special trees.}
We now define some specific trees which play crucial roles in the tree balance literature.

We start with the \emph{caterpillar tree} or simply \emph{caterpillar}, denoted by $\cat_n$. It is the rooted binary tree with $n$ leaves that fulfills that either $n=1$, or $n \geq 2$ and then, additionally, the tree has exactly one cherry. It is well known that the caterpillar is the only rooted binary tree in which each inner vertex is adjacent to at least one leaf.

For technical reasons, in the following we sometimes also need to consider the \emph{empty tree} $T^\emptyset$. $T^\emptyset$ is the unique tree with $n=0$ and $h=0$, which consists of no vertices and no edges.

Next, the \emph{maximally balanced tree}, denoted by $\mb_n$, is the rooted binary tree with $n$ leaves in which all inner vertices are \emph{balanced}, i.e., we have $|n_{v_a}-n_{v_b}|\leq 1$ for all $v\in\Vint(T)$, where $v_a$ and $v_b$ denote the children of $v$. It is known for $T_n^{mb}$ that its standard decomposition is $T_n^{mb}=\left(T^{mb}_{\left\lceil \frac{n}{2}\right\rceil},T^{mb}_{\left\lfloor \frac{n}{2}\right\rfloor}\right)$. In particular, for $T_n^{mb}$ we have $n_a=\left\lceil \frac{n}{2}\right\rceil$ and $n_b=\left\lfloor \frac{n}{2}\right\rfloor$, where $n_a$ and $n_b$ denote the sizes of the maximum pending subtrees of $T_n^{mb}$ with $n_a\geq n_b$. 

We continue with the \emph{greedy from the bottom tree}, denoted by $\gfb_n$. It is the rooted binary tree with $n$ leaves that results from greedily clustering trees of minimal leaf numbers, starting with $n$ single vertices and proceeding until only one tree is left as described by~\cite[Algorithm 2]{coronado2020}.

We now consider a special tree that only exists for numbers of leaves $n$ which are a power of two, i.e., in case we have $n=2^h$ for some $h \in \mathbb{N}$. This tree is the \emph{fully balanced tree of height $h$}, denoted by $\fb_h$, which is the rooted binary tree with $n=2^h$ leaves in which all leaves have depth precisely $h$. Note that for $h \geq 1$, we have $\fb_h = \left(\fb_{h-1}, \fb_{h-1}\right)$. Moreover, for $h \in \mathbb{N}_{\geq 0}$, $\fb_h= \mb_{2^h} = \gfb_{2^h}$. 

Another tree that is of high relevance in our manuscript is the \emph{binary echelon tree} or \emph{echelon tree} for short, which is denoted by $T_n^{be}$. This tree was recursively defined in \cite{Currie2024} as follows: If $n=0$, $T_n^{be}$ is the empty tree $T^{\emptyset}$,  and if $n=1$, $T_n^{be}$ consists of a single vertex. Now, for $n>1$, let $k_n=\left\lceil \log_2n \right\rceil$. Then, $T_n^{be}$ contains a maximal pending subtree $T_a$ which equals $T_{k_n-1}^{fb}$ and a maximal pending subtree $T_b$ which equals $T_{n-2^{k_n-1}}^{be}$. Thus, the binary echelon tree is uniquely defined (up to isomorphism) for all $n$. See Figure~\ref{fig:moreminima} for an exemplary depiction of $T_9^{be}$ and a comparison to $T_9^{gfb}$.

Last but not least, note that all special trees introduced here -- $T_n^{mb}$, $T_n^{gfb}$, $T_n^{be}$, $T_{k_n}^{fb}$ and $T_n^{cat}$ -- have the property that for each of their inner vertices $v\in \mathring{V}$, the induced pending subtree $T_v$ is a special tree of the same kind (e.g., all induced pending subtrees of a caterpillar tree are also caterpillar trees, etc.). \\

\noindent\textbf{ Tree balance concepts.}
In phylogenetics, but also in other research areas, tree balance plays a vital role. Generally, balance and imbalance indices are distinguished. However, note that every balance index can be converted into an imbalance index by inverting its sign (and vice versa). 

In order to define (im)balance indices, however, we first need to consider tree shape statistics: A \emph{tree shape statistic} is a function $t: \bt_n\rightarrow \mathbb{R}$. Based on this, a tree imbalance index is defined as follows: A tree shape statistic $t$ is called an \emph{imbalance index} if and only if
	\begin{enumerate}[(i)]
		\item the caterpillar $\cat_n$ is the unique tree maximizing $t$ on its domain $\bt_n$ for all $n \geq 1$ and
		\item the fully balanced tree $\fb_h$ is the unique tree minimizing $t$ on $\bt_n$ for all $n = 2^h$ with $h \in \mathbb{N}$.
	\end{enumerate}

    Similarly, a tree shape statistic $t$ is called a \emph{balance index} if and only if
	\begin{enumerate}[(i)]
		\item the caterpillar $\cat_n$ is the unique tree minimizing $t$ on its domain $\bt_n$ for all $n \geq 1$ and
		\item the fully balanced tree $\fb_h$ is the unique tree maximizing $t$ on $\bt_n$ for all $n = 2^h$ with $h \in \mathbb{N}$.
	\end{enumerate}

Following, for instance, \cite{Fischer2023}, we call a balance or imbalance index \emph{local} if it fulfills the following criterion: If two trees differ only in one rooted subtree, the difference between their indices is equal to the difference between the indices of these differing subtrees \cite{Fischer2023,Mir2013}. More precisely, a balance or imbalance index $t$ is local, if
    $t(T)-t(T')=t(T_v)-t(T_v')$
    for $T$, $T'\in\bt_n$ with $T'$ obtained from $T$ by replacing the pending subtree $T_v$ by a subtree $T_v'$ of the same size.  
Following \cite{Fischer2023,Matsen2007}, a \emph{binary recursive tree shape statistic} of length $x$ is an ordered pair $(\lambda,r)$ consisting of a vector $\lambda\in\mathbb{R}^x$ and an $x$-vector $r$ of symmetric functions $r_i:\mathbb{R}^x\times\mathbb{R}^x\rightarrow \mathbb{R}$, also called \emph{recursions}. Here, $x$ is called \emph{length} of $r$ and denotes the number of recursions. Moreover, $\lambda$ with $\lambda_i = r_i(T)$ for $T$ with $n=1$ contains the base case for each recursion, for $i=1,\ldots,x$. Note that even though each recursion $r_i$ operates on 2 vectors of length $x$, namely $(r_1(T_a),\ldots,r_x(T_a))$ and $(r_1(T_b),\ldots,r_x(T_b))$, that contain real numbers, following \cite{Fischer2023} we use the shorthand $r_i(T_a,T_b)$ instead of $r_i((r_1(T_a),\ldots,r_x(T_a)),(r_1(T_b),\ldots,r_x(T_b)))$. Here, $T=(T_a,T_b)$ is the standard decomposition of a binary tree $T\in\bt_n$.

As stated before, more than 25 (families of) balance and imbalance indices are known in the literature \cite{Fischer2023}. However, in this manuscript, we mainly focus on two of the most famous ones, namely on the Sackin and Colless index, both of which are imbalance indices.

\begin{itemize}
\item The \emph{Sackin index} \cite{Sackin1972, Fischer2021, Fischer2023} of a rooted binary tree $T$ is defined as 
    \[S(T) = \sum\limits_{v\in\Vint(T)}n_v=\sum\limits_{x \in V^1(T)}\delta_x.\] Note that the two sums can easily be  shown to be equal (cf. \cite{Fischer2021}). Also note that whenever $n=0$, i.e., when we consider the empty tree $T^\emptyset$ with $\mathring{V}(T^\emptyset)=\emptyset$, the sum of the Sackin index is empty. Thus, by convention, in this case we have $S(T^\emptyset)=0$.

    Using the convention $n_{v_a}\geq n_{v_b}$ for all $v\in \mathring{V}(T)$, the Sackin index can also be re-written as:  $$S(T)=\sum\limits_{v\in \mathring{V}(T)}n_v=\sum\limits_{v\in\mathring{V}(T)} (n_{v_a}+n_{v_b})=\sum\limits_{v\in\mathring{V}(T)} n_{v_a}+\sum\limits_{v\in\mathring{V}(T)} n_{v_b}=N_a(T)+N_b(T),$$ where for each $v\in \mathring{V}(T)$, $v_a$ and $v_b$ denote the children of $v$, and where $N_a(T):=\sum\limits_{v\in\mathring{V}(T)}n_{v_a}$ and $N_b(T):=\sum\limits_{v\in\mathring{V}(T)}n_{v_b}$. 
\item The \emph{Colless index} \cite{Colless1990,Fischer2023} of a rooted binary tree $T$ is defined as \[C(T) = \sum\limits_{v\in\Vint(T)}|n_{v_a}-n_{v_b}|,\]
where $n_{v_a}$ and $n_{v_b}$ denote the sizes of the maximal pending subtree of $v \in \mathring{V}(T)$, respectively. Using the convention $n_{v_a}\geq n_{v_b}$ for all $v\in \mathring{V}(T)$, this can be re-written as:  $$C(T)=\sum\limits_{v\in\mathring{V}(T)} (n_{v_a}-n_{v_b})=\sum\limits_{v\in\mathring{V}(T)} n_{v_a}-\sum\limits_{v\in\mathring{V}(T)} n_{v_b}=N_a(T)-N_b(T),$$ where $N_a(T)$ and $N_b(T)$ are as defined above. Note again that whenever $n=0$, i.e., when we consider $T^\emptyset$, the sum of the Colless index is empty, implying that we have $C(T^\emptyset)=0$.
\end{itemize}

In this manuscript, we will analyze the sums $N_a$ and $N_b$ more in-depth. We will show that while $N_a$ is an imbalance index, $N_b$ is a balance index. In this context, the difference $C(T)-S(T)$ will also play an important role (and turn out to be an imbalance index, too). Therefore, we now also introduce the following two notions:

\begin{itemize}
\item The $\Delta_{CS}$ index of a rooted binary tree $T$ is defined as: $\Delta_{CS}(T)=C(T)-S(T)$, and
\item the $\Delta_{SC}$ index of a rooted binary tree $T$ is defined as: $\Delta_{SC}(T)=S(T)-C(T)=-\Delta_{CS}(T)$.
\end{itemize}

Note that we already here refer to $\Delta_{CS}$ and $\Delta_{SC}$ as indices rather than mere tree shape statistics, because we will show subsequently that they are indeed (im)balance indices. 

Next, we turn our attention to some results known from the literature.

\subsection{Prior results}

In this subsection, we remind the reader of some results from the literature which will be used later on to derive our own results. The results as presented here are mainly taken from \cite{Fischer2023}, but can also be found, for instance, in \cite{Fischer2021} (for the Sackin index), in \cite{coronado2020} (for the Colless index and $T_n^{gfb}$) and in \cite{cleary2025} (for $T_n^{gfb}$).

\subsubsection{Known results on \texorpdfstring{$T_n^{gfb}$}{Tgfb}}\label{sec:prelim_gfb}
In our manuscript, we will exploit the following decomposition of $T_n^{gfb}$, which is based on a related result from \cite{coronado2020}, but which -- in the exact form we are using it in -- was more recently stated in \cite{cleary2025}:

For $n \geq 2$, we let $T_n^{gfb}=(T_a,T_b)$, where $n_a$ and $n_b$ denote the sizes of $T_a$ and $T_b$, respectively. Let $k_n=\lceil\log_2(n)\rceil$. Then, we have:
\begin{enumerate}
\item If $2^{k_n-1} < n \leq 3\cdot 2^{k_{n}-2}$, then $n_a=n-2^{k_n-2}$, $n_b=2^{k_n-2}$ and $T_b=T_{k_n-2}^{fb}$.
\item If $3 \cdot 2^{{k_{n}}-2}\leq n\leq 2^{k_{n}}$, then $ n_a=2^{k_{n}-1}$, $n_b=n-2^{k_{n}-1}$ and $T_a=T_{k_n-1}^{fb}$.
\end{enumerate}

\subsubsection{Known results on \texorpdfstring{$T_n^{be}$}{Tbe}}\label{sec:prelim_be}
The recursive definition of the echelon tree already shows that in $T_n^{be}=(T_a^{be},T_b^{be})$, where $T_a^{be}$ has $n_a=2^{k_{n}-1}$ many leaves and $T_b^{be}$ has $n_b=n-2^{k_{n}-1}$ many leaves, we have $(n_a,n_b)=(2^{k_{n}-1},n-2^{k_{n}-1})$, where $k_n=\lceil \log_2 n \rceil$. We will refer to this pair $(n_a,n_b)$ as the \emph{echelon partition} of $n$. 
     It is important to note that when $n=2^{k_n}$, we have $T_n^{be}=T_n^{mb}=T_n^{gfb}=T_{k_n}^{fb}$.
Furthermore, note that it has already been observed in the literature that the echelon tree is a so-called rooted binary weight tree  \cite{Currie2024} as introduced in \cite{Kersting2021}. However, a non-recursive construction of the echelon tree has so far not been stated in the literature, and as for most values of $n$ there are several rooted binary weight trees with $n$ leaves, it is not obvious which one coincides with the echelon tree. We will therefore formally introduce rooted binary weight trees and give an explicit construction for the echelon tree in Section~\ref{sec:echelon_results}.

\subsubsection{Known results on the Sackin index} \label{sec:prelim_sackin}

The \emph{Sackin index} is an imbalance index, and it was shown in \cite{Fischer2021} that $S(T)\geq -2^{k_n}+n\cdot(k_n+1)$, where $k_n=\lceil \log_2(n) \rceil$, and that the given bound is tight. The trees that achieve this  minimum have been fully characterized in \cite{Fischer2021}, and it was shown there that the set of minimal trees consists precisely of the ones in which each pair of leaves $x$ and $y$ has the property that $|\delta_x-\delta_y|\leq 1$. In particular, this implies that for all $n \in \mathbb{N}$, $\mb_n$ and $\gfb_n$ are contained in the set of Sackin-minimal trees. Moreover, from Section~\ref{sec:prelim_gfb} it can be easily seen that if a tree  
$T=(T_a, T_b)$, with $T_a\in \bt_{n_a}$ and $T_b\in \bt_{n_b}$, is a minimal Sackin tree with $n$ leaves, then ${n}_a - {n}_b\leq n^{gfb}_a-n^{gfb}_b$, where $n^{gfb}_a$ and $n^{gfb}_b$ denote the leaf numbers of the maximal pending subtrees of $T_n^{gfb}$, respectively, with $n^{gfb}_a \geq n^{gfb}_b$. This is due to the fact that otherwise, we could find two leaves $x$, $y$ in $T$ with $|\delta_x-\delta_y|\geq 2$. Also, for each pair $(n_a,n_b)$ with $n_a+n_b=n$ and $n_a-n_b\leq n_a^{gfb}-n_b^{gfb}$ there is a Sackin minimal tree $T$ with $T=(T_a,T_b)$ and $T_a$ of size $n_a$ and $T_b$ of size $n_b$ (as these minima can be constructed from attaching cherries to leaves of $T_{k_n-1}^{fb}$, cf. \cite{Fischer2021,Fischer2023}). We therefore refer to such pairs $(n_a,n_b)$ with $n_a+n_b=n$ and $n_a-n_b\leq n_a^{gfb}-n_b^{gfb}$ as \emph{Sackin partitions} of $n$.

Note that for $\fb_n$, we have $S(\fb_{k_n})=2^{k_n}\cdot k_n$, and $\fb_{k_n}$ is the unique minimum of the Sackin index for $n=2^{k_n}$.

Concerning the maximum, for every binary tree $T \in \bt_n$, the Sackin index fulfills
\[
S(T) \leq \frac{n\cdot (n+1)}{2}-1.
\]
This bound is tight for all \( n \in \mathbb{N}_{\geq 1} \) and is achieved precisely by \( \cat_n \). 

Moreover, for every binary tree \( T \in \bt_n \), the Sackin index \( S(T) \) can be computed in time \( O(n) \), and the Sackin index is local. For $T=(T_a,T_b)$ with $n$ leaves, the Sackin index fulfills $S(T)=S(T_a)+S(T_b)+n$ and is thus also recursive  \cite{Fischer2021,Fischer2023}. 

\subsubsection{Known results on the Colless index}\label{sec:prelim_Colless}

The Colless index is an imbalance index \cite{coronado2020,Fischer2023}. Several explicit formulas to express $c_n$ for each value of $n\in \mathbb{N}$, where $c_n$ denotes the minimum value of the Colless index of a tree in $\bt_n$, have been established \cite{hamoudi2017,coronado2020}, but we will exemplarily use only the following one: Let $\breve{s}(x)$ denote the distance from $x\in\mathbb{R}$ to its nearest integer, which implies $\breve{s}(x)=\min_{z \in \mathbb{Z}}|x-z|$. Then, \[ C(T)\geq \sum\limits_{j=1}^{k_n-1} 2^j\cdot \breve{s}(2^{-j}\cdot n), \] where $k_n=\lceil\log_2(n)\rceil$. The bound is tight for all $n\in\mathbb{N}_{\geq 1}$, and it equals  $C(\fb_{k_n})=0$ for $n=2^k$ for all $k\in \mathbb{N}$. The trees that achieve the minimum for general values of $n$ have been fully characterized in \cite{coronado2020}, and it was in particular shown there that for each $n$, $\mb_n$ and $\gfb_n$ are contained in the set of Colless-minimal trees. As it is somewhat complex, we refrain from re-stating the characterization of all Colless minima here, but we point out that it is based on a full characterization of the set $QB(n)$ of pairs $(n_a,n_b)$ with $n_a\geq n_b$ and $n_a+n_b=n$ such that $T=(T_a,T_b)$ -- where $T_a$ is a Colless minimal tree on $n_a$ leaves and $T_b$ is a Colless minimal tree on $n_b$ leaves -- is also Colless minimal. The set $QB(n)$ will also play a role in our manuscript, and we will refer to pairs $(n_a,n_b) \in QB(n)$ as \emph{Colless partitions} of $n$. Note that it was shown in \cite{coronado2020} that if $(n_a,n_b)$ is a Colless partition,  then it is also a Sackin partition. Moreover, it was shown in the same manuscript that all Colless-minimal trees are also Sackin-minimal (but not vice versa). 
    
    Note that the Sackin index is an upper bound for the Colless index. In particular, we have $C(T)\leq S(T)$ for all trees $T \in \bt_n$, with equality if and only if $n=1$. The latter is true as in case $T$ consists only of a single leaf, we have $C(T)=S(T)=0$, as both indices by definition sum up over all inner vertices, but since $T$ has no inner vertex, these sums are both empty. Moreover, when $n>1$, we may denote for each $v \in \mathring{V}(T)$ the children of $v$ with $v_a$ and $v_b$ such that $n_{v_a}\geq n_{v_b}\geq 1$. Then, we have $C(T)=\sum\limits_{v\in \mathring{V(T)}}n_{v_a}-n_{v_b}< \sum\limits_{v\in \mathring{V(T)}}n_{v_a}+n_{v_b}=S(T)$, where the inequality is true as $n_{v_b}\geq 1$ for all $v$. This shows that indeed $C(T) \leq S(T)$.
\par\vspace{0.3cm}
Concerning the maximum, for every binary tree $T\in\bt_n$, the Colless index fulfills \[ C(T)\leq \frac{(n-1)(n-2)}{2}.  \] This bound is tight for all $n\in\mathbb{N}_{\geq 1}$ and is achieved precisely by $\cat_n$.  

Moreover, for every binary tree $T\in\bt_n$, the Colless index $C(T)$ is local and can be computed in time $O(n)$. Additionally, for a tree $T=(T_a,T_b)$, where $n_a\geq n_b$, the Colless index fulfills $C(T)=C(T_a)+C(T_b)+n_a-n_b$ and is thus a recursive tree shape statistic \cite{Fischer2023,coronado2020}.

\section{Results} \label{sec:results}

The main goal of this section is twofold: First, we want to show that the Sackin and Colless indices are merely compound imbalance indices as they are linear combinations of $N_a$ and $N_b$, which we will show to be (im)balance indices of their own. Along the way, we will also investigate $\Delta_{CS}$ and $\Delta_{SC}$,
 which will also turn out to be (im)balance indices as they are closely related to $N_b$. Second, we will show that the echelon tree plays an important role for three of these new (im)balance indices, and we will investigate its properties further in Section \ref{sec:echelon_results}. We start our investigation with $\Delta_{CS}$, $\Delta_{SC}$ and $N_b$.   

\subsection{ \texorpdfstring{$\Delta_{CS}$, $\Delta_{SC}$ and $N_b$: three new (im)balance indices on $\bt_n$}{DeltaCS, DeltaSC and Nb: three new (im)balance indices on BTnstar}}

The main aim of this subsection is to prove that $\Delta_{CS}=C-S$ is an imbalance index and  both $N_b$ and $\Delta_{SC}$ are  balance indices on $\bt_n$. We will show this by exploiting their close relationship, which we will also elaborate on.

\begin{theorem}\label{thm:c-s_imbalance} Let $n\in \mathbb{N}$. Then, $\Delta_{CS}$ is an imbalance index, and $N_b$ as well as $\Delta_{SC}$ are balance indices on $\bt_n$.
\end{theorem}

Before we can prove this main result, we need some preliminary results. We begin with the following lemma, which provides a simple but important property of $\Delta_{CS}$. In fact, it establishes the close link between $\Delta_{CS}$ and $N_b$. 

\begin{lemma}\label{lem:D=sumnb} Let $T \in \bt_n$. Then, we have: $\Delta_{CS}(T)=-2\sum\limits_{v\in \mathring{V}(T)} n_{v_b}=-2N_b(T)$, where $n_{v_b}$ denotes the size of the \emph{smaller} pending subtree adjacent to $v$, i.e., if $v_a$ and $v_b$ denote the children of $v$, we have $n_{v_b}\leq n_{v_a}$. 
\end{lemma}

\begin{proof} Let $T \in \bt_n$. Then, for all $v \in \mathring{V}(T)$ we may assume without loss of generality that for the children $v_a$ and $v_b$ of $v$ we have  $n_{v_a}\geq n_{v_b}$. Note that this way, we can drop the absolute value in the definition of the Colless index. Thus, by definition and using the fact that $n_v=n_{v_a}+n_{v_b}$, we have:

\begin{align*}
\Delta_{CS}(T)=C(T)-S(T) 
&= \left(\sum\limits_{v \in \mathring{V}(T)} n_{v_a} - \sum\limits_{v \in \mathring{V}(T)} n_{v_b}\right) - \left(\sum\limits_{v \in \mathring{V}(T)} n_{v_a} + \sum\limits_{v \in \mathring{V}(T)} n_{v_b}\right) \\
&= -2 \sum\limits_{v \in \mathring{V}(T)} n_{v_b}=-2N_b(T).
\end{align*}
This completes the proof.

\end{proof}

\begin{remark}\label{rem:SCandNb} Note that Lemma~\ref{lem:D=sumnb} implies that $\Delta_{CS}$ is maximized whenever $N_b$ is minimized and vice versa. Since $\Delta_{CS}=-\Delta_{SC}$, it also shows that $\Delta_{SC}=2N_b$, implying that $\Delta_{SC}$ and $N_b$ have precisely the same minima and maxima. Therefore, in the following we mainly focus on $\CS$. We only directly mention $\SC$ and $N_b$ in the main theorems or whenever we give specific values.
\end{remark}

Next, we derive a recursion for $\Delta_{CS}$. This will be one of the most important ingredients needed in the analysis of the minimum value of $\Delta_{CS}$ as well as the trees that achieve it.

\begin{proposition}\label{prop:rec} Let $n \in \mathbb{N}_{\geq 2}$ and $T\in \bt_n$ with standard decomposition $T=(T_a,T_b)$, where $n_a$, $n_b$ denote the number of leaves of $T_a$ and $T_b$, respectively, such that $n_a\geq n_b$. Then, we have: 

$$\Delta_{CS}(T)=\Delta_{CS}(T_a)+\Delta_{CS}(T_b)-2n_b.$$
\end{proposition}

\begin{proof} Basically, $\Delta_{CS}$
 inherits its recursive property from $C$ and $S$ as follows: 

 \begin{align*}
 \Delta_{CS}(T)&=C(T)-S(T) \\
 &=(C(T_a)+C(T_b)+n_a-n_b)-(S(T_a)+S(T_b)+n) \\
 &= (C(T_a)-S(T_a))+(C(T_b)-S(T_b))+n_a-n_b-\underbrace{(n_a+n_b)}_{=n}\\
 &= \Delta_{CS}(T_a)+\Delta_{CS}(T_b)-2n_b.
 \end{align*}
\end{proof}

Proposition~\ref{prop:rec} immediately leads to the following important corollary. 

\begin{corollary}\label{cor:T_min_then_Ta_Tb_min}
    Let $T\in\bt_n$ with standard decomposition $T=(T_a,T_b)$ minimize $\Delta_{CS}$ in $\bt_n$. Then, $T_a$ minimizes $\Delta_{CS}$ in $\bt_{n_a}$ and $T_b$ minimizes $\Delta_{CS}$ in $\bt_{n_b}$.
\end{corollary}
\begin{proof}
 Let $T\in\bt_n$ with standard decomposition $T=(T_a,T_b)$ minimize $\Delta_{CS}$ in $\bt_n$ and assume $T_a$ is not minimal in $\bt_{n_a}$, implying there is a tree $T_a'\in\bt_{n_a}$ with $\Delta_{CS}(T_a')<\Delta_{CS}(T_a)$. Furthermore, let $T'=(T_a',T_b)$. Thus, by Proposition~\ref{prop:rec},
 \begin{align*}
     \Delta_{CS}(T) = \Delta_{CS}(T_a)+\Delta_{CS}(T_b)-2n_b>\Delta_{CS}(T_a')+\Delta_{CS}(T_b)-2n_b= \Delta_{CS}(T'),&
 \end{align*} a contradiction to the minimality of $T$.
 The case for $T_b$ works analogously. This completes the proof.
\end{proof}

\subsubsection{Maximizing \texorpdfstring{$\Delta_{CS}$}{DeltaCS} and minimizing \texorpdfstring{$N_b$}{Nb} as well as \texorpdfstring{$\Delta_{SC}$}{DeltaSC}}

 It is the aim of this subsection to show that $\cat_n$ is the unique tree maximizing $\Delta_{CS}$ and to derive the maximum value of $\Delta_{CS}$.

 \begin{proposition}\label{prop:D_cat}
Let $T\in\bt_n$. Then, we have $\Delta_{CS}(T) \leq -2n+2$, $\Delta_{SC}(T) \geq 2n-2$ and $N_b(T)\geq n-1$. These bounds are tight and are all precisely achieved if $T=\cat_n$.
 \end{proposition}
\begin{proof} We start by considering the caterpillar tree. $\cat_n$ is by definition the unique tree in which each inner vertex is adjacent to a leaf. Following the convention of $n_a\geq n_b$, we thus have $n_{v_b}=1$ for each $v\in \mathring{V}(\cat_n)$. Thus, $N_b(\cat_n)=\sum_{v\in \mathring{V}(\cat_n)} n_{v_b}=|\mathring{V}(\cat_n)|=n-1$, where the latter equality stems from the fact that $\cat_n$ is a rooted binary tree with $n$ leaves. However, in every non-caterpillar tree $T$ there is at least one inner node $v$ which is \emph{not} adjacent to a leaf, implying that at least one summand in $N_b(T)=\sum_{v\in \mathring{V}(T)} n_{v_b}$ is strictly larger than 1. All other summands are at least one, and the number of summands is again $n-1$. Hence, in total, we get $N_b(T)=\sum_{v\in \mathring{V}(T)} n_{v_b}>n-1=\sum\limits_{v\in \mathring{V}(\cat_n)} n_{v_b}=N_b(\cat_n)$. This completes the proof for $N_b$.

Concerning $\Delta_{CS}$, recall that by Lemma~\ref{lem:D=sumnb} we know $\Delta_{CS}(T)=-2N_b(T)$ and that $\Delta_{CS}(T)$ is maximal if and only if $N_b(T)$ is minimal, which in turn is the case if and only if $T$ equals $\cat_n$. Moreover, we have $\Delta_{CS}(\cat_n)=-2N_b(\cat_n)=-2\cdot(n-1)=-2n+2$. The statement for $\Delta_{SC}$ now follows by $\Delta_{CS}=-\Delta_{SC}$. This completes the proof.
\end{proof}

We now turn our attention to the more intricate minimum of $\Delta_{CS}$.

\subsubsection{Minimizing \texorpdfstring{$\Delta_{CS}$}{DeltaCS} and maximizing \texorpdfstring{$N_b$}{Nb} as well as  \texorpdfstring{$\Delta_{SC}$}{DeltaSC}}

The main aim of this section is to give a full characterization of all trees that minimize $\Delta_{CS}$. An important step in this direction will be to show that every tree that minimizes the Colless index also minimizes $\Delta_{CS}=C-S$. Note that this is highly non-obvious, because every Colless minimum also minimizes the Sackin index, cf. Section~\ref{sec:Prelim}, so the term $S(T)$ that is subtracted from $C(T)$ also gets minimal in this case. However, the following recurrence will be  an essential tool used to tackle the minimum in the following.

\begin{proposition}\label{prop:d_rec}
    Let $d_n$ be the minimal $\Delta_{CS}$-value for a tree in $\bt_n$. Then, we have $d_0=d_1=0$ and  
    \begin{align*}
        &d_{2n} =2d_n-2n\\
        &d_{2n+1} = d_{n+1}+d_n-2n
    \end{align*} for all $n\in\mathbb{N}_{\geq 1}$.
\end{proposition}
\begin{proof}
   Before proving the recurrences for $d_{2n}$ and $d_{2n+1}$ for $n\in\mathbb{N}_{\geq 1}$, we briefly show that the initial values $d_0=d_1=0$ are correct. Note that in $\bt_0$ and $\bt_1$ there is only one tree, respectively, namely $T^{\emptyset}$ and the tree consisting of only one vertex. In both cases, we have no inner vertices in the tree, implying that the sums defining $C(T)$ and $S(T)$ are empty, which shows that indeed we have $d_0=d_1=0-0=0$. 
   
\par\vspace{0.5cm}
   
    We now turn our attention to the recurrences. To show that $d_{2n} =2d_n-2n$ and $d_{2n+1} = d_{n+1}+d_n-2n$ for all $n\geq 1$, we perform a proof by induction on $n$. In the base case $n=1$, we need to calculate $d_{2\cdot 1}=d_2$ and $d_{2\cdot 1+1}=d_3$. However, note that there is only one tree each in $\bt_2$ and $\bt_3$ ($\cat_2$ and  $\cat_3$, respectively), which come with Sackin indices $S(\cat_2)=2$ and $S(\cat_3)=5$ and Colless indices $C(\cat_2)=0$ and $C(\cat_3)=1$. Thus, we have:  $$d_2=\Delta_{CS}(\cat_2)=C(\cat_2)-S(\cat_2)=0-2=2\cdot d_1-2\cdot 1,$$  as well as  $$d_3=\Delta_{CS}(\cat_3)=C(\cat_3)-S(\cat_3)=1-5=-4=d_2+d_1-2\cdot 1.$$

This completes the proof of the base case.
    
    We proceed with the induction step, i.e., we assume that the recurrences hold for all trees with $\tilde{n}<n$ leaves and now  consider trees with $n$ leaves. We  first show that $d_{2n} \leq 2d_{n}-2n$ and $d_{2n+1} \leq d_{n+1}+d_{n}-2n$ for all $n\in\mathbb{N}_{\geq 1}$.
    By Proposition~\ref{prop:rec} and Corollary~\ref{cor:T_min_then_Ta_Tb_min},  we have: 
    \begin{equation}\label{eq:d_2n_geq}
        \begin{aligned}
            d_{2n} &= \min\{d_{n_a}+d_{n_b}-2n_b|n_a \geq n_b\geq 1,n_a+n_b = 2n\}\\
            &=\min\limits_{j\in\{1,\ldots,n\}}\{d_{2n-j}+d_j-2j\} \leq 2d_{n}-2n, 
        \end{aligned}
    \end{equation}

    where the latter inequality can easily be verified by considering $j=n$.

Similarly, we have: 
    \begin{equation}\label{eq:d_2n+1_geq}
        \begin{aligned}
            d_{2n+1} &= \min\{d_{n_a}+d_{n_b}-2n_b|n_a \geq n_b\geq 1,n_a+n_b = 2n+1\}\\&= \min\limits_{j\in\{1,\ldots,n\}}\{d_{2n+1-j}+d_j-2j\}
            \leq d_{n+1}+d_{n}-2n,
        \end{aligned}
    \end{equation}

    where again the latter inequality can be verified by considering $j=n$. 
    
    To obtain equality in the recurrences, it remains to show that $d_{2n} \geq 2d_{n}-2n$ and $d_{2n+1} \geq d_{n+1}+d_{n}-2n$ holds as well. 
    
    First, consider $d_{2n}$. By applying the induction hypothesis to $d_{2n-j}$ and $d_j$, we get 
    \begin{equation}\label{eq:d_2n}
        \begin{aligned}
            d_{2n} &=\min\limits_{j\in\{1,\ldots,n\}}\{d_{2n-j}+d_j-2j\}\\
            &=\min\limits_{j\in\{1,\ldots,n\}}\left.\begin{cases}
               2\overbrace{\left(d_{n-\frac{j}{2}}+d_{\frac{j}{2}}-2\frac{j}{2}\right)}^{=:A}-2n & \text{for $j$ even}\\
                \underbrace{\left(d_{n-\frac{j-1}{2}}+d_{\frac{j-1}{2}}-2\left(\frac{j-1}{2}\right)\right)}_{=:B}+ \underbrace{\left(d_{n-\frac{j+1}{2}}+d_{\frac{j+1}{2}}-2\left(\frac{j+1}{2}\right)\right)}_{=:C}-2n+2 & \text{for $j$ odd}
            \end{cases}\right\}.
        \end{aligned}
    \end{equation}

    We want to show that $A$, $B$ and $C$ are all at least as large as $d_n$, except for $C$ in case $j$ is odd and $j=n$, in which case we have $C\geq d_n-2$. This will show that in all cases we have $2A-2n\geq 2d_n-2n$ and $B+C-2n+2\geq 2d_n-2n$.  In order to see this, we first  re-write the minimum defining $d_n$ as follows: 

\begin{equation}\label{eq:d_n}
    \begin{aligned}
          d_{n} &= \min\{d_{n_a}+d_{n_b}-2n_b|n_a\geq n_b\geq 1, n_a+n_b=n\}= \min\left\{d_{n-i}+d_i-2i|1\leq i\leq \left\lfloor\frac{n}{2}\right\rfloor\right\}.
    \end{aligned} 
    \end{equation}

Note that in Equation~\eqref{eq:d_n}, $i$ ranges from 1 (which is the minimum size of $n_b$) to $\left\lfloor \frac{n}{2}\right\rfloor$, due to $n_b\leq n_a$ and $n_a+n_b=n$, implying $n_b\leq \frac{n}{2}$. The floor function then stems from the fact that $n_b$ can only assume integer values, and it ensures that the bound is tight regardless if $n$ is even or odd.

We now use Equation~\eqref{eq:d_n} to analyze $A$, $B$, and $C$.  \begin{itemize} \item We start with term $A$, i.e.,  we consider the case that $j\in \{1,\ldots,n\}$ and $j$ is even, implying $j\geq 2$. In this case, we thus have $1\leq \frac{j}{2}\leq \lfloor\frac{n}{2}\rfloor$ (where the floor function again is due to the the fact that $j$ is an integer), and thus, for every even  $j$ in $\{2,\ldots,n\}$, we have $A=d_{n-\frac{j}{2}}+d_{\frac{j}{2}}-2\frac{j}{2}\geq d_n$ by Equation~\eqref{eq:d_n}. Note that this in particular implies that if $j$ is even, we have  $2A-2n \geq 2d_n-2n$. 

\item We continue with terms $B$ and $C$, which only play a role when $j$ is odd, so we consider the case that $j\in \{1,\ldots,n\}$ and $j$ is odd. We now consider three subcases.

    \begin{itemize}
    \item We first analyze the case $j=1$. In this case, obviously $B=d_n$, and $C= d_{n-1}+d_1-2$. Thus, it is obvious that $ C\geq d_n$ by Equation~\eqref{eq:d_n} (which can be seen using $i=1$). 
\item If $j$ is odd and $j=n$, then, in particular, $n$ is odd, too. In this case, using $i=\frac{j-1}{2}=\frac{n-1}{2}=\left\lfloor \frac{n}{2}\right\rfloor$ (where the latter equality is true as $n$ is odd), we see, by Equation~\eqref{eq:d_n}, that $B\geq d_n$. Moreover, for $C$ we get: \begin{align*}C&= \underbrace{d_{n-\frac{n-1}{2}}+d_{\frac{n-1}{2}}-2\left(\frac{n-1}{2}\right)}_{\overset{\mbox{\tiny Eq.~\eqref{eq:d_n}}}{\geq}d_n} -2 \geq d_n-2.\end{align*} 
\item Lastly, if $j$ is odd and $3\leq j\leq n-1$, then $1\leq \left\lfloor\frac{j}{2}\right\rfloor=\frac{j-1}{2}<\frac{j+1}{2}=\left\lceil\frac{j}{2}\right\rceil\leq\left\lceil \frac{n-1}{2}\right\rceil \leq\left\lfloor\frac{n}{2}\right\rfloor$. Thus, we can use $i=\frac{j-1}{2}$ in Equation~\eqref{eq:d_n} to derive that $B\geq d_n$. Similarly, we can use $i=\frac{j+1}{2}$ in Equation~\eqref{eq:d_n} to derive that $C\geq d_n$. 
\end{itemize}
Hence, in case $j$ is odd and $1\leq j\leq n$, we have $B+C-2n+2\geq 2d_n-2n$.
\end{itemize}

In summary, since the minimum in Equation~\eqref{eq:d_2n} is taken over values $2A-2n$ and $B+C-2n+2$, which all are at least as large as $2d_n-2n$, we can conclude $d_{2n}\geq 2d_n-2n$, which completes the proof for $d_{2n}$.

It remains to consider $d_{2n+1}$.
    Note that if $j$ is even, then $2n+1-j$ is odd and vice versa.
    Hence, by applying the induction hypothesis to $d_{2n+1-j}$ and $d_j$, we obtain
    \begin{equation}\label{eq:d_2n+1}
        \begin{aligned}
            d_{2n+1} &= \min\limits_{j\in\{1,\ldots,n\}}\{d_{2n+1-j}+d_j-2j\}\\
            &=\min\limits_{j\in\{1,\ldots,n\}}\left.\begin{cases}
                \overbrace{\left(d_{n+1-\frac{j}{2}}+d_{\frac{j}{2}}-2\frac{j}{2}\right)}^{=:D}+\overbrace{\left(d_{n-\frac{j}{2}}+d_{\frac{j}{2}}-2\frac{j}{2}\right)}^{=:E}-2n & \text{for $j$ even}\\
                 \underbrace{\left(d_{n-\frac{j+1}{2}+1}+d_{\frac{j+1}{2}}-2\left(\frac{j+1}{2}\right)\right)}_{=:F}+\underbrace{\left(d_{n-\frac{j-1}{2}}+d_{\frac{j-1}{2}}-2\left(\frac{j-1}{2}\right)\right)}_{=:G}-2n& \text{for $j$ odd}
            \end{cases}\right\}.
        \end{aligned}
    \end{equation}

    In the following, it is our aim to show that both $D+E$ as well as $F+G$ are at least large as $d_{n+1}+d_n$. In order to see this, first recall the representation of $d_n$ in Equation~\eqref{eq:d_n}. Analogously, we have for $d_{n+1}$: 
     \begin{equation}\label{eq:d_n+1}
    \begin{aligned}
          d_{n+1} &= \min\{d_{n_a}+d_{n_b}-2n_b|n_a\geq n_b\geq 1, n_a+n_b=n+1\}\\
          &= \min\left\{d_{n+1-i}+d_i-2i|1\leq i\leq \left\lfloor\frac{n+1}{2}\right\rfloor\right\}.
    \end{aligned} 
    \end{equation}

Note that in Equation~\eqref{eq:d_n+1}, $i$ ranges from 1 (which is the minimum size of $n_b$) to $\left\lfloor \frac{n+1}{2}\right\rfloor$, due to $n_b\leq n_a$ and $n_a+n_b=n+1$, implying $n_b\leq \frac{n+1}{2}$. The floor function then stems from the fact that $n_b$ can only assume integer values, and it ensures that the bound is tight regardless if $n$ is even or odd.

We now use Equation~\eqref{eq:d_n+1} to analyze $D$, $E$, $F$, and $G$. 

\begin{itemize}
\item We start with terms $D$ and $E$, i.e.,  we consider the case in which $j\in \{1,\ldots,n\}$ is even. In this case, we have $j\geq 2$ and thus $\frac{j}{2}\geq 1$, as well as $\frac{j}{2}\leq \frac{n}{2}\leq \left\lfloor \frac{n+1}{2}\right\rfloor$. We now consider $D$ and $E$ separately.

\begin{itemize} \item As we have $1\leq \frac{j}{2} \leq \left\lfloor \frac{n+1}{2}\right\rfloor$, we can apply Equation~\eqref{eq:d_n+1} using $i=\frac{j}{2}$ to conclude that $D\geq d_{n+1}$.

\item Next, if $n$ is even, we have $1\leq \frac{j}{2} \leq \left\lfloor \frac{n+1}{2}\right\rfloor=\left\lfloor \frac{n}{2}\right\rfloor$. 

If, however, $n$ is odd, we have 
$j\leq n-1$ (as $j\leq n$, where $n$ is odd and $j$ is even), and thus $\frac{j}{2}\leq \frac{n-1}{2}\leq \left\lfloor \frac{n}{2} \right\rfloor$. Therefore, in both cases we can apply Equation~\eqref{eq:d_n} to conclude that $E\geq d_n$.

\end{itemize}

\item We continue with terms $F$ and $G$, meaning we now consider the case in which $j$ is odd. We distinguish two cases.

\begin{itemize}
\item If $j=1$, we have $F=d_n+d_1-2$, which, by Equation~\eqref{eq:d_n+1} is at least $d_{n+1}$ (which can be seen using $i=1$). Thus, in this case we have $F\geq d_{n+1}$. 
Moreover, if $j=1$, we immediately get $G=d_n$.
\item If $j>1$ and $j$ is odd, we have $1<\frac{j}{2}<\frac{j+1}{2}\leq \frac{n+1}{2}$. Furthermore, we either have $j\leq n-1$ if $n$ is even or $\frac{n+1}{2}=\left\lfloor \frac{n+1}{2}\right\rfloor$ if $n$ is odd. In both cases $1\leq \frac{j+1}{2}\leq\lfloor\frac{n+1}{2}\rfloor$ and we can apply Equation~\eqref{eq:d_n+1} using $i=\frac{j+1}{2}$ to conclude that $F\geq d_{n+1}$. The fact that $\frac{j+1}{2}\leq\lfloor\frac{n+1}{2}\rfloor$ also gives us $1\leq \frac{j-1}{2}=\frac{j+1}{2}-1\leq \lfloor\frac{n+1}{2}\rfloor-1 \leq \lfloor\frac{n}{2}\rfloor.$ Therefore, we can apply Equation~\eqref{eq:d_n} using $i= \frac{j-1}{2}$ to $G$ to conclude that $G\geq d_n$.
\end{itemize}

\end{itemize}

So in all cases, we have $D\geq d_{n+1}$ and $E\geq d_n$, implying that $D+E-2n \geq d_{n+1}+d_n-2n$. Similarly, we have in all cases that $F\geq d_{n+1}$ and $G\geq d_n$, implying that $F+G-2n \geq d_{n+1}+d_n-2n$. In summary, using Equation~\eqref{eq:d_2n+1}, this shows that $d_{2n+1}\geq d_{n+1}+d_n-2n$, which completes the proof.
\end{proof}

The following theorem provides a family of $\Delta_{CS}$ minimal trees, namely all Colless minima. As a bonus, we can then easily derive the minimal value for $\Delta_{CS}$ and the maximal values for $N_b$ and $\Delta_{SC}$ from the minimal Colless and Sackin values. 

\begin{theorem} \label{thm:Dmin} Let $n\in \mathbb{N}$ and $T\in\bt_n$. Moreover, let $\breve{s}(x)=\min_{z \in \mathbb{Z}}|x-z|$ for $x\in\mathbb{R}$. Let $k_n=\lceil \log_2(n) \rceil$.  Then, we have: \begin{itemize}\item  $\Delta_{CS}(T) \geq C(\mb_n)-S(\mb_n)=\left(\sum\limits_{j=1}^{k_n-1} 2^j\cdot \breve{s}(2^{-j}\cdot n)\right)+2^{k_n}-n\cdot(k_n+1)$,
\item  $\Delta_{SC}(T) \leq S(\mb_n)-C(\mb_n)=-2^{k_n}+n\cdot(k_n+1)-\left(\sum\limits_{j=1}^{k_n-1} 2^j\cdot \breve{s}(2^{-j}\cdot n)\right)$ 
and
\item $N_b(T)\leq N_b(T_n^{mb})=-\frac{1}{2}\cdot \left(\left(\sum\limits_{j=1}^{k_n-1} 2^j\cdot \breve{s}(2^{-j}\cdot n)\right)+2^{k_n}-n\cdot(k_n+1)\right).$
\end{itemize} All three bounds are tight. Moreover, the bounds are achieved (among others) by every tree which is Colless minimal, i.e., by every tree $T$ for which we have $C(T)=\sum\limits_{j=1}^{k_n-1} 2^j\cdot \breve{s}(2^{-j}\cdot n)$. 
\end{theorem}
Our proof strategy for $\Delta_{CS}$ will be as follows: We first consider $\Delta_{CS}(T_n^{mb})$ and show that this term follows the same recurrences as the ones given by Proposition~\ref{prop:d_rec}. This proves that $\Delta_{CS}(T_n^{mb})=d_n$ for all $n\in \mathbb{N}$. As $T_n^{mb}$ is known to be Colless-minimal (cf. Section~\ref{sec:prelim_Colless}), this immediately shows that \emph{all} Colless-minimal trees also achieve the minimal value $d_n$ of $\Delta_{CS}$, because obviously all Colless minima have the same Colless value. Also, it is known from the literature that all Colless minima are also Sackin minima  (cf. Section~\ref{sec:prelim_Colless}). As the minimal values both for Colless and Sackin are also already known (cf. Sections~\ref{sec:prelim_sackin} and \ref{sec:prelim_Colless}), this automatically gives us a value for $d_n=\Delta_{CS}(T_n^{mb})$.

\begin{proof}[Proof of Theorem \ref{thm:Dmin}]
We start by considering $\Delta_{CS}(T_n^{mb})$ and prove that $\Delta_{CS}(T_n^{mb})=d_n$ by induction over $n$. For $n=0$ and $n=1$ we have $S(T_n^{mb})=C(T_n^{mb})=0$ and thus $\Delta_{CS}(T_n^{mb})=C(T_n^{mb})-S(T_n^{mb})=0-0=0=d_n$ for $n\in\{0,1\}$. 

Now, let $n\geq 1$ and assume $\CS(T_{n'}^{mb}) = d_{n'}$ for all $n' < n$. By Proposition~\ref{prop:rec} we know that: $$\Delta_{CS}(T_n^{mb})=\Delta_{CS}(T_{n_a}^{mb})+\Delta_{CS}(T_{n_b}^{mb})-2n_b.$$ Moreover, we know that for the maximally balanced tree we have  $n_a=\left\lceil \frac{n}{2}\right\rceil$ and $n_b=\left\lfloor \frac{n}{2}\right\rfloor$ (cf. Section~\ref{prelim:def}), leading to:

\begin{align*}\Delta_{CS}(T_n^{mb})&=\Delta_{CS}\left(T_{\left\lceil \frac{n}{2}\right\rceil}^{mb}\right)+\Delta_{CS}\left(T_{\left\lfloor \frac{n}{2}\right\rfloor}^{mb}\right)-2\left\lfloor \frac{n}{2}\right\rfloor\\ &\overset{\left\lceil \frac{n}{2}\right\rceil,\left\lfloor \frac{n}{2}\right\rfloor<n}{=} d_{\left\lceil \frac{n}{2}\right\rceil} + d_{\left\lfloor \frac{n}{2}\right\rfloor} -2\left\lfloor \frac{n}{2}\right\rfloor
\overset{\text{Prop.~\ref{prop:d_rec}}}{=}
d_n.
\end{align*}

This shows that indeed, $T_n^{mb}$, which is a known Colless minimal tree, also minimizes $\Delta_{CS}$. As stated above, every Colless minimum is also a Sackin minimum, and the minimum values for Colless and Sackin are known (and identical for all trees assuming these values). Thus, we derive for every Colless minimal tree $T$ the term $d_n=\Delta_{CS}(T)=C(T)-S(T)=C(T_n^{mb})-S(T_n^{mb})$. Inserting the terms for the minimal values of the Sackin and Colless indices given by Sections~\ref{sec:prelim_sackin} and \ref{sec:prelim_Colless} completes the statement for $\Delta_{CS}$. 

The respective statement  for $\Delta_{SC}$ now follows by $\Delta_{SC}=-\Delta_{CS}$, and the one for $N_b$ follows by using $\Delta_{CS}=-2N_b$, which we know from Lemma~\ref{lem:D=sumnb}. This completes the proof.
\end{proof}

Hence, by Theorem~\ref{thm:Dmin}, we know that all Colless minimal trees also minimize $\Delta_{CS}$. However, these trees are not the only ones leading to these extrema. In order to derive a characterization of all $\Delta_{CS}$ minimal trees, an important step is given by the following proposition, which shows that the echelon tree $T_n^{be}$ also plays another important role for $\Delta_{CS}$, even though it is not generally an extremal tree neither for Sackin nor for Colless.

\begin{proposition}\label{prop:echelon} For all $n\in\mathbb{N}$, the echelon tree $T_n^{be}$ minimizes $\Delta_{CS}$ (and maximizes $N_b$ and $\Delta_{SC}$). 
\end{proposition}

\begin{proof} Our proof strategy is to show that for all $n$, we have $\Delta_{CS}(T_n^{be})=\Delta_{CS}(T_n^{gfb})$. The rest of the statement then follows by Theorem~\ref{thm:Dmin} and Lemma~\ref{lem:D=sumnb} as well as Remark~\ref{rem:SCandNb}.

We show $\Delta_{CS}(T_n^{be})=\Delta_{CS}(T_n^{gfb})$ by proof of induction over $n$. First, consider the base case $n=1$. Obviously, $T_1^{be}$ is equivalent to $T_1^{gfb}$, and thus, $\Delta_{CS}(T_1^{be})=\Delta_{CS}(T_1^{gfb})$.

Now, let $n>1$ and assume that $\Delta_{CS}(T_{n'}^{be})=\Delta_{CS}(T_{n'}^{gfb})$ for all $n'<n$. Consider $T_n^{be}$ and let let $k_n=\lceil \log_2n\rceil$, in the following. Moreover, we denote by $T_n^{be}=(T_a^{be},T_b^{be})$ the standard decomposition of $T_n^{be}$ into its two maximal pending subtrees $T_a^{be}$ and $T_b^{be}$ with $n_a^{be}=2^{k_n-1}$ and $n_b^{be}$ (with $n_a^{be}\geq n_b^{be}$) leaves, respectively. Similarly, we denote by $T_n^{gfb}=(T_a^{gfb},T_b^{gfb})$ the standard decomposition of $T_n^{gfb}$ into its two maximal pending subtrees $T_a^{gfb}$ and $T_b^{gfb}$ with $n_a^{gfb}$ and $n_b^{gfb}$ (with $n_a^{gfb}\geq n_b^{gfb}$) leaves, respectively.

Note that by definition of $k_n$, we have $n \in \{2^{k_n-1}+1,\ldots,2^{k_n}\}$. Based on this, we now distinguish two cases based on the two possible sizes of $n_a^{gfb}$ and $n_b^{gfb}$ as stated in Section~\ref{sec:prelim_gfb}.

\begin{enumerate}
\item If $n\geq 3\cdot 2^{k_n-2}$, we have $n_a^{gfb}=2^{k_n-1}$ (cf. Section~\ref{sec:prelim_gfb}), which in turn equals $n_a^{be}$ by definition of the echelon tree. Thus, as $n_a^{gfb}+n_b^{gfb}=n_a^{be}+n_b^{be}=n$, we must also have $n_b^{gfb}=n_b^{be}=n-2^{k_n-1}$. As $T_a^{gfb}=T_a^{be}=T_{k_n}^{fb}$, we clearly have $\Delta_{CS}(T_a^{be})=\Delta_{CS}(T_a^{gfb})$. Moreover, as $n_b^{gfb}=n_b^{be}<n$ and by the induction hypothesis, we know that $\Delta_{CS}(T_b^{be})=\Delta_{CS}(T_b^{gfb})$ as both employ the same number of leaves. This leads to the following:

\begin{align*}
\Delta_{CS}(T_n^{be}) & \overset{\text{\tiny Prop.~\ref{prop:rec}}}{=} \Delta_{CS}(T_a^{be})+\Delta_{CS}(T_b^{be})-2n_b^{be}\\ &= \Delta_{CS}(T_a^{gfb})+\Delta_{CS}(T_b^{gfb})-2n_b^{gfb}
\overset{\text{\tiny Prop.~\ref{prop:rec}}}{=} \Delta_{CS}(T_n^{gfb}),
\end{align*}

where we know the latter term to be minimal by Theorem~\ref{thm:Dmin}. This completes the proof for the case $n \geq 3\cdot 2^{k_n-2}$.

\item If $n< 3\cdot 2^{k_n-2}$, we know by Section~\ref{sec:prelim_gfb} that $T_b^{gfb}=T_{k_n-2}^{fb}$, i.e., $n_b^{gfb}=2^{k_n-2}$ and $n_a=n-2^{k_n-2}$. However, $T_a^{be}=T_{k_n-1}^{fb}$ and $n_b^{be}=n-2^{k_n-1}$. $T_n^{gfb}$ and $T_n^{be}$ are schematically depicted in Figure~\ref{fig:echelon_opt}, together with a third tree, namely $T'$, which we consider here, too. Formally, $T'=(T_a',T_b')$ is defined to be the tree with $T_a'=\left(T_{n-2^{k_n-1}}^{be},T_{k_n-2}^{fb}\right)$ and $T'_b=T_{k_n-2}^{fb}$. Note that $T'$ is such that its root implies the gfb partition and one of its maximum pending subtrees coincides with the gfb tree's smaller one. Additionally, its larger maximum pending subtree's two maximum pending subtrees together with this smaller subtree are identical to the three subtrees employed by the echelon tree, cf. Figure \ref{fig:echelon_opt}. Thus, $T'$ can be regarded as a mixture of $T_n^{gfb}$ and $T_n^{be}$.

\begin{figure}[ht]
  \centering
  \includegraphics[scale=0.9]{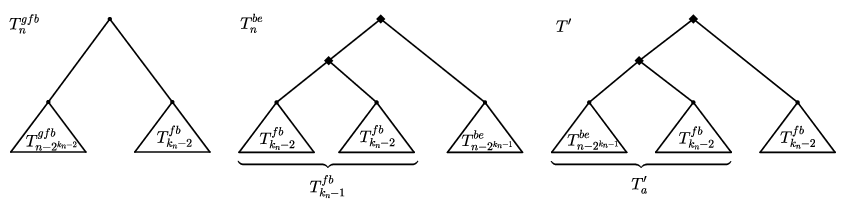}
  \caption{Trees $T_n^{gfb}$, $T_n^{be}$ and $T'$ as used in the proof of Proposition~\ref{prop:echelon}. Subtrees are depicted as triangles. Note that the diamond nodes highlight the nodes which differ between $T_n^{be}$ and $T'$. }
\label{fig:echelon_opt}
\end{figure}

Our proof strategy now is as follows. We first show that $T_a'$ is an echelon tree of size $n-2^{k_n-2}$ and that we therefore must have $\Delta_{CS}(T_n^{gfb})=\Delta_{CS}(T')$. We then show that $\Delta_{CS}(T')=\Delta_{CS}(T_n^{be})$, completing the proof.

We start by showing that $T_a'$ is an echelon tree.
Note that $n-2^{k_n-1}+2^{k_n-2}=n-2^{k_n-2}<2^{k_n-1}$ as $n<3\cdot 2^{k_n-2}$ by assumption. This shows that $2^{k_n-2}$ is the largest power of two fitting into the $n-2^{k_n-2}$ leaves of $T_a'$, and with $T_a'=\left( T_{k_n-2}^{fb},T_{n-2^{k_n-1}}^{be}\right)$ it is clear that $T_a'$ is indeed an echelon tree.

As $n_a'=n-2^{k_n-2}<n$, $\Delta_{CS}(T_a')=\Delta_{CS}(T_{n-2^{k_n-2}}^{gfb})$ by the induction hypothesis. Using Proposition~\ref{prop:rec}, this leads to:
\begin{align*}
\Delta_{CS}(T') \ &\overset{\text{\tiny Prop.~\ref{prop:rec}}}{=}\Delta_{CS}(T_a')+\Delta_{CS}(T_b')-2n_b'\\&= \Delta_{CS}(T_{n-2^{k_n-2}}^{gfb})+\Delta_{CS}(T_{k_n-2}^{fb})-2\cdot 2^{k_n-2}
\\&= \Delta_{CS}(T_{n-2^{k_n-2}}^{gfb})+\Delta_{CS}(T_{2^{k_n-2}}^{gfb})-2\cdot 2^{k_n-2}\\&\overset{\text{\tiny Prop.~\ref{prop:rec}}}{=}\Delta_{CS}(T_n^{gfb}).
\end{align*}

It remains to show that $\Delta_{CS}(T_n^{gfb})=\Delta_{CS}(T')=\Delta_{CS}(T_n^{be})$.

We now explicitly compare $\Delta_{CS}(T')=C(T')-S(T')$ with $\Delta_{CS}(T_n^{be})=C(T_n^{be})-S(T_n^{be})$ and show that $\Delta_{CS}(T')=\Delta_{CS}(T_n^{be})$. In order to make this comparison, consider Figure~\ref{fig:echelon_opt} again and note that $T_n^{be}$ and $T'$ only differ in the two vertices highlighted by diamonds. This immediately shows that:

\begin{align*}C(T')&=C(T_n^{be})-(2^{k_n-1}-(n-2^{k_n-1}))-(2^{k_n-2}-2^{k_n-2})\\
&\ +(n-2^{k_n-1}+2^{k_n-2}-2^{k_n-2})+|n-2^{k_n-1}-2^{k_n-2}| \\
&\overset{\star}{=}C(T_n^{be})+2n-2^{k_n}-2^{k_n-1}+(2^{k_n-2}-(n-2^{k_n-1}))\\
&=C(T_n^{be})+n-2^{k_n}+2^{k_n-2}\\
S(T')&=S(T_n^{be})-2^{k_n-1}+(n-2^{k_n-1}+2^{k_n-2})\\
&= S(T_n^{be})-2^{k_n}+n+2^{k_n-2},
\end{align*}

where the equality marked with the asterisk stems from the fact that we have already seen that $2^{k_n-2}$ is the largest power of two fitting into $n-2^{k_n-1}+2^{k_n-2}=n-2^{k_n-2}$, showing that $T_{k_n-2}^{fb}$ forms the larger maximum pending subtree in $T_{n-2^{k_n-1}+2^{k_n-2}}^{be}=T_{n-2^{k_n-2}}^{be}$ by the definition of echelon trees.

Now we can calculate $\Delta_{CS}(T')$:

\begin{align*}
\Delta_{CS}(T')&=C(T')-S(T')\\ &= \underbrace{(C(T_n^{be})+n-2^{kn}+2^{k_n-2})}_{=C(T_n^{be})}-\underbrace{(S(T_n^{be})-2^{k_n}+n+2^{k_n-2})}_{=S(T_n^{be})}=\Delta_{CS}(T_n^{be}).
\end{align*}

Hence, altogether we have $\Delta_{CS}(T_n^{gfb})=\Delta_{CS}(T')=\Delta_{CS}(T_n^{be})$. This completes the proof.
\end{enumerate}
\end{proof}

In order to show that $\Delta_{CS}$, $\Delta_{SC}$ and $N_b$ are indeed (im)balance indices, we need the following lemma.

\begin{lemma}\label{lem:more_than_echelon} Let $T=(T_a,T_b)$ be a rooted binary tree with $n=n_a+n_b$ leaves with $n_a\geq n_b$ being the numbers of leaves in $T_a$ and $T_b$, respectively. Let $k_n=\lceil \log_2 n\rceil $. Then, if $n_a>2^{k_n-1}$, $T$ is not $\Delta_{CS}$ minimal.
\end{lemma}

\begin{proof} Seeking a contradiction, assume there is a tree $T$ as described in the lemma which is indeed $\Delta_{CS}$ minimal. We consider the smallest $n$ for which such a tree $T$ exists. 
As $T$ minimizes $\Delta_{CS}$ by assumption, by Corollary~\ref{cor:T_min_then_Ta_Tb_min}, $T_a$ and $T_b$ must minimize $\Delta_{CS}$, too. We now consider tree $\widetilde{T}$ as depicted in Figure~\ref{fig:notworsethanechelon}, which is like $T$, but in which $T_a$ got exchanged by $T_{n_a}^{be}$. 

\begin{figure}[ht]
  \centering
  \includegraphics[scale=1.2]{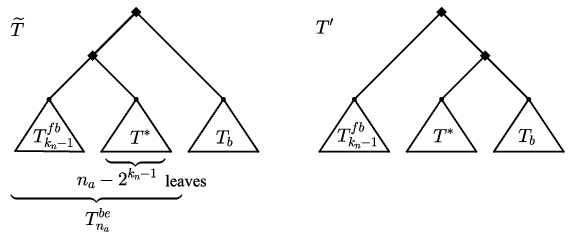}
  \caption{ Trees $\widetilde{T}$ and $T'$ as used in the proof of Lemma~\ref{lem:more_than_echelon}. Note that both trees employ the same pending subtrees and the only vertices in which they differ are highlighted with diamonds.}
\label{fig:notworsethanechelon}
\end{figure}

By Proposition~\ref{prop:echelon}, $T_{n_a}^{be}$ is, just like $T_a$, $\Delta_{CS}$ minimal, implying $\Delta_{CS}(T_a)=\Delta_{CS}(T_{n_a}^{be})$. Furthermore, as both trees share the same subtree $T_b$, we in summary have (using Proposition~\ref{prop:rec}) that $\Delta_{CS}(T)=\Delta_{CS}(\widetilde{T})$. However, note that $2^{k_n}\geq n>n_a>2^{k_n-1}$ by assumption, so clearly $2^{k_n-1}$ is the largest power of two fitting into $n_a$. Therefore, $T_{n_a}^{be}$ consists of a larger maximal pending subtree $T_{k_n-1}^{fb}$ and a smaller subtree $T^\ast$ of size $n_a-2^{k_n-1}$, cf. Figure~\ref{fig:notworsethanechelon}. Thus, we have $2^{k_n-1}>n_a-2^{k_n-1}$. 

We now consider $T'$, which is like $\widetilde{T}$ except that $T^\ast$ is moved from the larger maximum pending subtree of $\widetilde{T}$ to the smaller one, see again Figure~\ref{fig:notworsethanechelon}. Our aim now is to show that $\Delta_{CS}(T')<\Delta_{CS}(\widetilde{T})=\Delta_{CS}(T)$, which contradicts the assumption on $T$ and thus completes the proof.

In order to see this, we now explicitly analyze $\Delta_{CS}(T')=C(T')-S(T')$. First, note that $\widetilde{T}$ and $T'$ only differ in the two nodes highlighted in Figure~\ref{fig:notworsethanechelon} by diamonds. This way, we easily see that:

\begin{align*}C(T')&=C(\widetilde{T})-(2^{k_n-1}-(n_a-2^{k_n-1}))-(n_a-n_b)\\
& \ \ + |(n_a-2^{k_n-1})-n_b|+(2^{k_n-1}-(n_a-2^{k_n-1}+n_b))\\
&= C(\widetilde{T}) -n_a+|n_a-2^{k_n-1}-n_b|\\
S(T')&=S(\widetilde{T})-n_a+(n_a-2^{k_n-1}+n_b)=S(\widetilde{T})-2^{k_n-1}+n_b,
\end{align*}

which directly shows:

\begin{align*}
\Delta_{CS}(T')&=C(T')-S(T')\\
&= C(\widetilde{T})-S(\widetilde{T})-n_a-n_b+2^{k_n-1}+|n_a-2^{k_n-1}-n_b|\\
&=\Delta_{CS}(\widetilde{T})-(n_a+n_b-2^{k_n-1})+\max\{\underbrace{n_a -2^{k_n-1}-n_b}_{<n_a+n_b-2^{k_n-1}},\underbrace{-n_a+2^{k_n-1}+n_b}_{\substack{
<n_a+n_b-2^{k_n-1} \\
\mbox{ \small as } n_a>2^{k_n-1}
}}\}\\
&< \Delta_{CS}(\widetilde{T}).
\tag{\theequation}\stepcounter{equation}\label{eq:deltaCS}
\end{align*}

Hence, we have $\Delta_{CS}(T')<\Delta_{CS}(\widetilde{T})$, which is clearly a contradiction to $\Delta_{CS}(\widetilde{T})=\Delta_{CS}(T)$ being minimal by assumption. This completes the proof.
\end{proof}

Before we can finally show that $\Delta_{CS}$ is an imbalance index, we still need to prove the following proposition.

\begin{proposition} \label{prop:D_fb}
Let $n=2^{k_n}$ for some $k_n\in \mathbb{N}$. Then, $\fb_{k_n}$ is the unique tree minimizing $\Delta_{CS}$ and maximizing both $N_b$ and $\Delta_{SC}$, and the minimal value of $\Delta_{CS}$ is $-2^{k_n}\cdot k_n$ and the maximal values of $\Delta_{SC}$ and $N_b$ are $2^{k_n}\cdot k_n$  and $2^{k_n-1}\cdot k_n$, respectively. 
\end{proposition}

\begin{proof}  Before we start proving uniqueness, note that we are guaranteed optimality of $T_k^{fb}$ for all $n=2^k$ by Theorem~\ref{thm:Dmin} as $T_k^{fb}$ is a Colless minimal tree. 

We now show uniqueness by proof of induction on $n$. In case $n=2^0=1$ or $n=2^1=2$, the only tree  coincides with $T^{fb}_0$ and $T^{fb}_1$, respectively.

 Now, let $k_n>1$ and assume that for every $k_{n'}<k_n$, $T^{fb}_{k_{n'}}$ is the unique minimizer of $\CS$ for trees with $2^{k_{n'}}$ leaves. Consider $T=(T_a,T_b)$ with $n=2^{k_n}$ leaves.
 Let $n_a$ and $n_b$ denote the numbers of leaves of $T_a$ and $T_b$, respectively, such that $n_a\geq n_b$. Together with $n_a+n_b=n$ this implies $n_a\geq \frac{n}{2}=2^{k_n-1}$. However, as $T$ minimizes $\Delta_{CS}$, we know from Lemma~\ref{lem:more_than_echelon} that $n_a\leq 2^{k_n-1}$. Thus, we must have $n_a=2^{k_n-1}$ and thus also $n_b=n-n_a=2^{k_n}-2^{k_n-1}=2^{k_n-1}$. 

Now, by Corollary~\ref{cor:T_min_then_Ta_Tb_min}, $T_a$ and $T_b$ must be $\Delta_{CS}$ minimal, too, and as $n_a$ and $n_b$ with  $n_a=n_b=2^{k_n-1}$ are powers of 2 that are smaller than $n$, by the induction hypothesis, $T_a$ and $T_b$ must both be fully balanced trees. Hence, $T_a=T_b=T_{k_n-1}^{fb}$. This shows that $T=(T_a,T_b)=(T_{k_n-1}^{fb},T_{k_n-1}^{fb})$ is a fully balanced tree. 

The value of $\Delta_{CS}(T_{k_n}^{fb})=C(T_{k_n}^{fb})-S(T_{k_n}^{fb})$ can be directly derived from the fact that $C(T_{k_n}^{fb})=0$ and $S(T_{k_n}^{fb})=2^{k_n}\cdot k_n$ (cf. Sections~\ref{sec:prelim_sackin} and \ref{sec:prelim_Colless}). The value for $\Delta_{SC}$ then follows by $\Delta_{SC}=-\Delta_{CS}$, and the one for $N_b$ is a consequence of Lemma~\ref{lem:D=sumnb}. This completes the proof. 
\end{proof}

Before we can characterize all $\Delta_{CS}$ minimal trees, we need one more lemma.

\begin{lemma}\label{lem:betweenBE&GFB}
Let $T=(T_a,T_b)$ be a rooted binary tree with $n=n_a+n_b$ leaves with $n_a\geq n_b$ being the numbers of leaves in $T_a$ and $T_b$, respectively. Let $k_n=\lceil \log_2(n)\rceil$ and let $(n_a^{gfb},n_b^{gfb})$ and $(n_a^{be}=2^{k_n-1},n_b^{be})$ denote the gfb and echelon partitions, respectively. Then, if $n_a^{gfb}<n_a<n_a^{be}=2^{k_n-1}$, $T$ is not $\Delta_{CS}$ minimal.
\end{lemma}

\begin{proof}
Seeking a contradiction, we assume there is a tree $T=(T_a,T_b)$ with $n_a^{gfb}<n_a<n_a^{be}=2^{k_n-1}$ that is $\CS$ minimal. First note that if $n \geq 3\cdot 2^{k_n-2}$, we know by Section \ref{sec:Prelim} that $n_a^{gfb}=2^{k_n-1}=n_a^{be}$, which shows that then $n_a^{gfb}<n_a<n_a^{be}$ cannot happen. So we must have $n<3\cdot 2^{k_n-2}$. 
Then, we have $n_b^{gfb}=2^{k_n-2}$ and thus $2^{k_n-1}>n_a>n_a^{gfb}\geq n_b^{gfb}=2^{k_n-2}>n_b$. This shows that $T_{n_a}^{be}$, the echelon tree with $n_a$ leaves, must contain $T_{k_n-2}^{fb}$ as a maximum pending subtree. We now construct tree $\widetilde{T}=(T_{n_a}^{be},T_b)$ as depicted in Figure~\ref{fig:characterization}. This tree has the same maximum pending subtree $T_b$ as $T$, but $T_a$ has been replaced by $T_{n_a}^{be}$. However, since $T$ is $\Delta_{CS}$ minimal, $T_a$ is $\Delta_{CS}$ minimal by Corollary~\ref{cor:T_min_then_Ta_Tb_min}. Furthermore, $T_{n_a}^{be}$ is also $\Delta_{CS}$ minimal by Proposition~\ref{prop:echelon} as it is an echelon tree. Therefore, we must have $\Delta_{CS}(T_a)=\Delta_{CS}(T_{n_a}^{be})$. This leads to:

\begin{align*}
\Delta_{CS}(T)&\overset{\text \tiny Prop.~\ref{prop:rec}}{=} \Delta_{CS}(T_a) + \Delta_{CS}(T_b) -2n_b = \Delta_{CS}(T_{n_a}^{be}) + \Delta_{CS}(T_b) -2n_b= \Delta_{CS}(\widetilde{T}).
\end{align*}

Thus, we know that $\Delta_{CS}(\widetilde{T})=\Delta_{CS}(T)$, implying $\widetilde{T}$ is $\Delta_{CS}$ minimal just like $T$.

\begin{figure}[ht]
  \centering
  \includegraphics[scale=1.2]{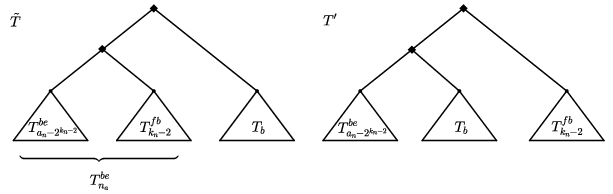}
  \caption{Trees $\widetilde{T}$ and $T'$ as needed in the proof of Lemma~\ref{lem:betweenBE&GFB}. Note that both trees employ the same pending subtrees and the only vertices in which they differ are highlighted with diamonds.}
\label{fig:characterization}
\end{figure}

We now consider tree $T'$ from Figure~\ref{fig:characterization}. $T'$ is like $\widetilde{T}$, but $T_b$ and $T_{k_n-2}^{fb}$ are interchanged. Note that $\widetilde{T}$ and $T'$ differ only in the two vertices marked with diamonds. This immediately leads to (calculations omitted as they are analogous to those of Lemma \ref{lem:more_than_echelon}): 

\begin{align*}C(T')&=C(\widetilde{T}) -2^{k_n}+n_a+2n_b+|n_a-n_b-2^{k_n-2}|\\
S(T')&=S(\widetilde{T})-2^{k_n-2}+n_b,
\end{align*}

which directly shows:
\begin{align*}
\Delta_{CS}(T')&=C(T')-S(T')\\&=\Delta_{CS}(\widetilde{T})+n_a+n_b-3\cdot 2^{k_n-2}+|n_a-n_b-2^{k_n-2}|
\\
&=\Delta_{CS}(\widetilde{T})+n_a+n_b-3\cdot 2^{k_n-2}+\max\{n_a-n_b-2^{k_n-2},-n_a+n_b+2^{k_n-2}\}\\
&=\Delta_{CS}(\widetilde{T})+\max\{2\underbrace{(n_a-2^{k_n-1})}_{<0},2\underbrace{(n_b-2^{k_n-2})}_{<0}\} < \Delta_{CS}(\widetilde{T}).
\end{align*}

Hence, we derive $\Delta_{CS}(T')<\Delta_{CS}(\widetilde{T})$, which contradicts the fact that, as $T$ is $\Delta_{CS}$ minimal by assumption and as we have $\Delta_{CS}(\widetilde{T})=\Delta_{CS}(T)$, $\widetilde{T}$ is $\Delta_{CS}$ minimal, too. 

\end{proof}

We continue with a full characterization of all $\Delta_{CS}$ minimal trees. Note that while we know from Theorem~\ref{thm:Dmin} that all Colless minima also minimize $\Delta_{CS}$, and by Proposition~\ref{prop:echelon} that the same is true for the echelon tree, these trees generally form only a subset of $\Delta_{CS}$ minimal trees, cf. Figure~\ref{fig:moreminima}. 

\begin{figure}[ht]
  \centering
  \includegraphics[scale=1.2]{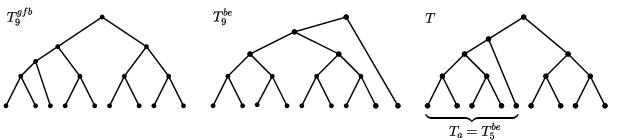}
  \caption{A depiction of all $\Delta_{CS}$ minimal trees for $n=9$: The set consists of $T_9^{gfb}$, $T_9^{be}$ and $T=(T_a,T_b)$. Note that at the root, $T$ has a Colless partition as $n_a=n_a^{gfb}=5$ and $n_b=n_b^{gfb}=4$. However, $T$ is not Colless minimal as it employs a non-Colless minimal subtree $T_a=T_5^{be}$. As this subtree is $\Delta_{CS}$ minimal, too, so is $T$. Note that $T$ can be regarded as a \enquote{combination} of the echelon minimality property and the Colless minimality property. Theorem~\ref{thm:characterization} shows that echelon trees, Colless minima and combinations thereof are the only $\Delta_{CS}$ minima.}
\label{fig:moreminima}
\end{figure}

However, it turns out that these trees are the basic building blocks of which all $\Delta_{CS}$ minimal trees consist. Our characterization of $\Delta_{CS}$ minimal trees is based on pairs $(n_a,n_b)$ which are either Colless partitions as defined in Section~\ref{sec:prelim_Colless} or echelon partitions as defined in Section~\ref{prelim:def}.

\begin{theorem} \label{thm:characterization} Let $n\in \mathbb{N}$. Let $T$ be a rooted binary tree with $n$ leaves. Then, we have: $T$ minimizes $\Delta_{CS}$ and maximizes $N_b$ and $\Delta_{SC}$ if and only if all pending subtrees $T_v$ rooted at inner nodes $v\in \mathring{V}(T)$ have maximal pending subtrees $T_v=(T_{{v_a}},T_{{v_b}}$) with $n_{v_a}\geq n_{v_b}$ many leaves, respectively, such that $(n_{v_a},n_{v_b})$ is either a Colless partition  or an echelon partition. 
\end{theorem}

\begin{proof}
We show the directions of the \enquote{if and only if}-statement separately.

\begin{enumerate}
\item We first want to show that every $\Delta_{CS}$ minimal tree has the property stated by the theorem. Seeking a contradiction, we consider a smallest $n$ for which such a $T=(T_a,T_b)$ exists which minimizes $\Delta_{CS}$ but which has an inner node that is neither the echelon nor a Colless partition. By the minimality of $T$, this node must be the root, which shows that $(n_a,n_b)$ is neither a Colless nor an echelon partition (where $n_a$, $n_b$ denote the numbers of leaves of $T_a$, $T_b$). Note that $T_a$ and $T_b$ must fulfill the conditions of the theorem as they are also $\Delta_{CS}$ minimal by Corollary~\ref{cor:T_min_then_Ta_Tb_min} and as they have fewer than $n$
 leaves. In the following, we again define $k_n=\lceil \log_2 n\rceil$. 

We consider three cases now based on $n_a$.

\begin{enumerate}
\item Consider the cases  $n_a>2^{k_n-1}=n_a^{be}$ or $n_a^{gfb}<n_a<n_a^{be}=2^{k_n-1}$. Then, by Lemma~\ref{lem:more_than_echelon} or Lemma~\ref{lem:betweenBE&GFB}, respectively, $T$ is not $\Delta_{CS}$ minimal, a contradiction. 

\item Note that the only remaining case is $n_a<n_a^{gfb}$. This is due to the fact that we have excluded all values of $n_a$ larger than $n_a^{gfb}$ except for $n_a^{be}$, but $n_a^{be}$ and $n_a^{gfb}$ are additionally excluded by the assumption that $(n_a,n_b)$ is neither a Colless nor an echelon partition. Thus, the only remaining case we need to analyze is the case $n_a<n_a^{gfb}$.

In this case, we know that $(n_a,n_b)$ is a Sackin partition, cf. Section~\ref{sec:prelim_sackin}. Since we assume $(n_a,n_b)$ is not a Colless partition, it must be true  that $(n_a,n_b)$ is a partition that has a corresponding Sackin minimal tree, but not a Colless minimal tree. We now construct a tree $T'=(T_a',T_b')$ with $n$ leaves as follows: $T_a'$ has $n_a'=n_a$ leaves and $T_b'$ has $n_b'=n_b$ leaves, and $T_a'$ and $T_b'$ are both Colless minima (and thus also Sackin minima). This way, we ensure that $T'$ is a Sackin minimum and that $T_a'$ and $T_b'$ are also $\Delta_{CS}$ minimal by Theorem~\ref{thm:Dmin}. 

As $T$ is $\Delta_{CS}$ minimal by assumption, $T_a$ and $T_b$ are $\Delta_{CS}$ minimal, too, by Corollary~\ref{cor:T_min_then_Ta_Tb_min}. As $n_a=n_a'$ and $n_b=n_b'$, this shows that $\Delta_{CS}(T_a)=\Delta_{CS}(T_a')$ and $\Delta_{CS}(T_b)=\Delta_{CS}(T_b')$.

Using Proposition~\ref{prop:rec}, we can conclude that $\Delta_{CS}(T)=\Delta_{CS}(T')$:

\begin{align*}\Delta_{CS}(T)&\overset{\text{\tiny Prop.~\ref{prop:rec}}}=\Delta_{CS}(T_a) +\Delta_{CS}(T_b) -2n_b\\ &= \Delta_{CS}(T_a') +\Delta_{CS}(T_b') -2n_b'=\Delta_{CS}(T').\end{align*}

Next, let $T''=(T_a'',T_b'')$ with $n=n_a''+n_b''$ leaves and $n_a''\geq n_b''$ be a Colless minimal tree. As by Theorem~\ref{thm:Dmin}, $T''$ is also $\Delta_{CS}$ minimal, we know $\Delta_{CS}(T)=\Delta_{CS}(T')=\Delta_{CS}(T'')$. We now show that the latter equality is a contradiction: 

\begin{align*}\Delta_{CS}(T')&=C(T')-S(T')> C(T'')-S(T')=C(T'')-S(T'')=\Delta_{CS}(T''),\end{align*}

where the inequality stems from the fact that $T''$ is a Colless minimum, whereas $T'$ cannot be one (as $(n_a',n_b')=(n_a,n_b)$ is not a Colless partition of $n$). The first equality sign after the inequality is due to the fact that $T'$ and $T''$ are both Sackin minimal trees, implying $S(T')=S(T'')$. However, this shows that the equality $\Delta_{CS}(T)=\Delta_{CS}(T'')$ as derived above is a contradiction.

\end{enumerate}

As all cases lead to contradictions, the assumption was wrong and such a tree $T$ cannot exist. This completes the first part of the proof. 

\item Next, we want to show that every tree that at each of its inner vertices has a pending subtree inducing a Colless or echelon partition is indeed $\Delta_{CS}$ minimal. We do so by proof of induction on $n$. In case $n=1$, the tree contains no inner vertices and there is nothing to show. There is exactly one tree $T$ with $n=2$ leaves and the only inner vertex $v\in \Vint(T)$ induces the pair $(n_{v_a},n_{v_b})=(1,1)$ which is both a Colless and an echelon partition. As $T$ is the only tree, it is obviously $\CS$ minimal.

Now, let $n>2$ and assume the statement holds for every tree with $n'<n$ leaves. Consider $T=(T_a,T_b)$ with $n$ leaves. Let $n_a$ denote the number of leaves of $T_a$ and $n_b$ the number of leaves of $T_b$ with $n_a\geq n_b$. Then, as $T$ has the property that all its inner vertices induce Colless or echelon partitions, so do $T_a$ and $T_b$. However, as $n_a,n_b<n$ and by the induction hypothesis, we know that $T_a$ and $T_b$ are in fact $\Delta_{CS}$ minimal. 
Now we choose $T_a'$ and $T_b'$ as follows: If $(n_a,n_b)$ is a Colless partition, we choose $T_a'$ and $T'_b$ to be Colless minimal trees with $n_a$ and $n_b$ leaves, respectively. If, on the other hand, $(n_a,n_b)$ is the echelon partition, we choose $T_a'$ and $T'_b$ to be echelon trees with $n_a$ and $n_b$ leaves, respectively. We now consider $T'=(T_a',T_b')$. As $T_a'$ and $T'_b$ are either Colless minima or echelon trees, they are $\Delta_{CS}$ minimal by Theorem~\ref{thm:Dmin} or Proposition~\ref{prop:echelon}, respectively. Hence, by the minimality of $T_a$ and $T_b$, we must have $\Delta_{CS}(T_a)=\Delta_{CS}(T_a')$ and $\Delta_{CS}(T_b)=\Delta_{CS}(T_b')$. This immediately leads to:

\begin{align*} 
\Delta_{CS}(T')&\overset{\text{\tiny Prop.~\ref{prop:rec}}}{=} \Delta_{CS}(T'_a) +\Delta_{CS}(T'_b)-2n_b'=\Delta_{CS}(T_a) +\Delta_{CS}(T_b)-2n_b=\Delta_{CS}(T).
\end{align*}

Thus, we have $\Delta_{CS}(T)=\Delta_{CS}(T')$. However, by construction $T'$ is either a Colless minimum or an echelon tree. Thus, we know that $T'$ is $\Delta_{CS}$ minimal, which in turn shows that $T$ is $\CS$ minimal. This completes the proof. 
\end{enumerate}

\end{proof}

We conclude this section with the following corollary, which gives a recursive formula to count the number of $\Delta_{CS}$ minima.

\begin{corollary}\label{cor:no_of_DeltaCS_min}
For each $n \in \mathbb{N}$, let $a(n)$ denote the number of $\Delta_{CS}$ minima (and $N_b$ as well as $\Delta_{SC}$ maxima). Then, $a(1)=1$ and for all $n\geq 2$, we have:

$$a(n)= \sum\limits_{\overset{(n_a,n_b) \in \mathcal{M}(n)}{n_a>n_b}}a(n_a)\cdot a(n_b) + \delta_{\text{\tiny even}}(n) \cdot \left(\binom{a\left(\frac{n}{2}\right)}{2}+a\left(\frac{n}{2}\right)\right),$$
where $\delta_{\text{\tiny even}}(n)=1$ if $n$ is even and 0 otherwise, and where $\mathcal{M}(n)$ is the set of Colless and echelon partitions $(n_a,n_b)$ of $n$. 
\end{corollary}

\begin{proof} The proof is a direct consequence of Theorem~\ref{thm:characterization} as all minimal trees are composed of two maximal pending subtrees which are also minimal, and these must have sizes that induce Colless or echelon partitions. In the even case, we need to avoid counting a combination of trees twice, which is why the simple sum is only considering the case $n_a>n_b$. In case $n_a=n_b=\frac{n}{2}$, on the other hand, every combination of two distinct trees is taken care of by the binomial coefficient, whereas every pair of equal trees is counted by $a\left(\frac{n}{2}\right)$, which explains the rightmost sum in the equation. This completes the proof. 
\end{proof}

We conclude this section by noting that the sequence $(a_n)_{n\geq 1}$ with $a_n=a(n)$, which starts with the values $\{1, 1, 1, 1, 2, 2, 1, 1, 3, 6, 6, 4, 4, 3, 1, 1, 4, 13, 25, 26\}$ was not contained in the Online Encyclopedia of Integer Sequences OEIS \cite{oeis}, which indicates that it might not have appeared in other contexts before. It has been submitted to the OEIS in the course of this manuscript \cite[Sequence A397642]{oeis}. Moreover, we want to mention that the reader can find more information on the set $\mathcal{M}(n)$ in \cite{coronado2020}, where all Colless partitions are characterized; so $\mathcal{M}(n)$ can simply be characterized by using those results and adding the echelon partition.

\subsubsection{Further properties of \texorpdfstring{$\Delta_{CS}$, $\Delta_{SC}$ and $N_b$}{DeltaCS, DeltaSC and Nb} }

In this subsection,  we discuss some properties of $\CS$, $\SC$ and $N_b$ which are generally considered desirable for tree (im)balance indices (cf. \cite{Fischer2023}). We start with the fact that, like most tree (im)balance indices\cite{Fischer2023}, $\CS$, $\SC$ and $N_b$ can be computed in linear time.

\begin{proposition}\label{prop:cs_linear_time}
    For every $T\in\bt_n$, $\CS(T)$, $\SC(T)$ and $N_b(T)$ can be computed in time $O(n)$.
\end{proposition}
\begin{proof}
As stated in Sections~\ref{sec:prelim_sackin} and \ref{sec:prelim_Colless}, both the Sackin and the Colless index can be computed in time $O(n)$. By definition, we have $\CS=C-S$ and $\SC=S-C$. Moreover, by Lemma~\ref{lem:D=sumnb}, we have $N_b=-\frac{1}{2}\Delta_{CS}=-\frac{1}{2}(C-S)$.  Thus, all three indices, $\CS$, $\SC$ and $N_b$, are  linear combinations of Sackin and Colless. Therefore, they, too, can be computed in time $O(n)$.
\end{proof}

\begin{proposition}\label{prop:cs_local}
    All three (im)balance indices $\CS$, $\SC$ and $N_b$ are local.
\end{proposition}
\begin{proof}
    Let $T$, $T'\in\bt_n$ with $T'$ obtained from $T$ by replacing the pending subtree $T_v$ by a subtree $T_v'$ of the same size. Using the locality of Sackin and Colless, we obtain
    \begin{align*}
    \CS(T)-\CS(T') &=(C(T)-S(T))-(C(T')-S(T'))\\
    &=(C(T_v)-C(T_v'))-(S(T_v)-S(T_v'))\\
    &=\CS(T_v)-\CS(T_v'). 
    \end{align*} This shows the locality of $\Delta_{CS}$. The respective statements for $\SC$ and $N_b$ follow from Lemma~\ref{lem:D=sumnb} and Remark~\ref{rem:SCandNb}, respectively.
\end{proof}

We now turn our attention to recursiveness. 

\begin{proposition}\label{prop:CS_brtss}
   $\CS$, $\SC$ and $N_b$ are binary recursive tree shape statistics.
\end{proposition}
\begin{proof}
    By Proposition~\ref{prop:rec}, for every $T\in\bt_n$ with standard decomposition $T=(T_a,T_b)$ and $n\geq 2$, we have: \[\CS(T) = \CS(T_a)+\CS(T_b)-(n_a+n_b)+|n_a-n_b|.\]  Furthermore, $\CS(T)=0$ for $T\in\bt_1$. Hence, $\CS$ can be written as a binary recursive tree shape statistic of length 2 as follows: 
    \begin{itemize}
        \item $\lambda =(0,1)$ as base case for the tree with $n=1$ leaf,
        \item for trees $T=(T_a,T_b)$ we have recursions $r_1(T_a,T_b) = \CS(T_a)+\CS(T_b)+|n_a-n_b|-(n_a+n_b)$ ($\CS$ value) and $r_2(T_a,T_b)=n_a+n_b$ (leaf numbers).
    \end{itemize}
    Obviously, $\lambda\in\mathbb{R}^2$ and $r_i:\mathbb{R}^2\times\mathbb{R}^2\rightarrow \mathbb{R}$, for $i=1,2$. Furthermore, $r_1,r_2$ are independent of the order of subtrees. 
   
    Analogously to Proposition~\ref{prop:rec}, we have \[\SC(T) = \SC(T_a)+\SC(T_b)+(n_a+n_b)-|n_a-n_b|\] and, by Lemma~\ref{lem:D=sumnb} and Proposition~\ref{prop:rec},
    \begin{align*}
            N_b(T) =-\frac{1}{2}\CS(T) &=-\frac{1}{2}\CS(T_a)-\frac{1}{2}\CS(T_b) +\frac{1}{2}((n_a+n_b)-|n_a-n_b|)\\
            &= N_b(T_a)+N_b(T_b)+ \frac{1}{2}((n_a+n_b)-|n_a-n_b|)  
    \end{align*}
    for $T\in\bt_n$ with $n\geq2$ and $\SC(T)=N_b(T)=0$ for $T\in\bt_1$.
    Thus, $\SC$ and $N_b$ can be expressed analogously to $\CS$ as binary recursive tree shape statistics of length 2.
\end{proof}

\subsection{ \texorpdfstring{$N_a$: another new imbalance index on $\bt_n$}{Na: another new imbalance index on BTnstar}}

In order to show that the Sackin and Colless indices are merely compound indices in the sense that they are  linear combinations of \enquote{more basic} (im)balance indices, we still need to analyze $N_a$. It is the main aim of this subsection to prove that $N_a$ is an imbalance index on $\bt_n$. It turns out that this statement can be derived from the respective fact for the Sackin index combined with our knowledge on $N_b$ from the previous section.

\begin{theorem}\label{thm:s-c_balance} Let $n\in \mathbb{N}$. Then, $N_a$ is an imbalance index on $\bt_n$. In particular, we have: \begin{itemize} 
\item $\cat_n$ is the unique maximizer of $N_a$, and we have $N_a(\cat_n)=\frac{n\cdot (n-1)}{2}$. 
\item If $n=2^{k_n}$ for some $k_n\in \mathbb{N}$, $T_{k_n}^{fb}$ is the unique minimizer of $N_a$, and we have $N_a(T_{k_n}^{fb})=2^{k_n-1}\cdot k_n =N_b(T_{k_n}^{fb})$.
\end{itemize}
\end{theorem}

\begin{proof}
Recall that $S(T)=N_a(T)+N_b(T)$, and thus $N_a(T)=S(T)-N_b(T)=S(T)+(-N_b(T))$ for $T\in\bt_n$. We know that the Sackin index is an imbalance index (cf. Section~\ref{sec:prelim_sackin}), and we know from Theorem~\ref{thm:c-s_imbalance} that $N_b$ is a balance index, which immediately shows that $-N_b$ is an imbalance index. Thus, $N_a$ is merely the sum of two imbalance indices and must therefore also be an imbalance index. Hence, it is clear that $\cat_n$ is the unique maximizer of $N_a$ for all $n$, and $T_{k_n}^{fb}$ is the unique minimizer of $N_a$ for $n=2^{k_n}$. 

In order to determine $N_a(T_n^{cat})$ and $N_a(T_{k_n}^{fb})$, we again use the identity $N_a=S-N_b$, which,  together with the respective values for $S$ from the literature (cf. Section~\ref{sec:prelim_sackin}) and for $N_b$ from Proposition~\ref{prop:D_cat} and Proposition~\ref{prop:D_fb}, leads to the following:

$$N_a(T_n^{cat})=S(T_n^{cat})-N_b(T_n^{cat})=\frac{n\cdot (n+1)}{2}-1-(n-1)=\frac{n\cdot (n-1)}{2},$$ and 

$$N_a(T_{k_n}^{fb})=S(T_{k_n}^{fb})-N_b(T_{k_n}^{fb})=2^{k_n}\cdot k_n-2^{k_n-1}\cdot k_n=2^{k_n-1}\cdot k_n\overset{\mbox{\tiny Prop.~\ref{prop:D_fb}}}{=} N_b(T_{k_n}^{fb}).$$
This completes the proof.
\end{proof}

Before we can continue, we need the following lemma.

\begin{lemma}\label{lem:Narec} Let $n\in \mathbb{N}_{\geq 2}$ and $T \in \bt_n$ with standard decomposition $T=(T_a,T_b)$, where $n_a$, $n_b$ denote the number of leaves of $T_a$ and $T_b$, respectively, such that $n_a\geq n_b$. Then, we have: 
$$N_a(T)=N_a(T_a)+N_a(T_b)+n_a.$$
\end{lemma}

\begin{proof} The statement can be easily seen using the definition of $N_a$ (using the convention that for every inner vertex $v$ of $T$ we have for its children $v_a$ and $v_b$ that $n_{v_a}\geq n_{v_b}$):

\begin{align*}
N_a(T)&=\sum\limits_{v\in\mathring{V}(T)}n_{v_a} = \sum\limits_{v\in\mathring{V}(T_a)}n_{v_a} + \sum\limits_{v\in\mathring{V}(T_b)}n_{v_a} + n_a\\
&=N_a(T_a)+N_a(T_b)+n_a.
\end{align*}
This completes the proof.
\end{proof}

The following corollary is a direct conclusion of Lemma~\ref{lem:Narec}.

\begin{corollary}\label{cor:NaOptsubtrees} Let $n\in \mathbb{N}$. Then, for a tree $T=(T_a,T_b)$ minimizing $N_a$ in $\bt_n$, we also have that $T_a$ and $T_b$ minimize $N_a$ in $\bt_{n_a}$ and $\bt_{n_b}$, respectively, where $n_a$ and $n_b$ are the numbers of leaves of $T_a$ and $T_b$, respectively. 
\end{corollary}

\begin{proof}
The proof of Corollary \ref{cor:NaOptsubtrees} is analogous to the proof of Corollary \ref{cor:T_min_then_Ta_Tb_min}, except that Lemma \ref{lem:Narec} is used instead of Proposition \ref{prop:rec}.
\end{proof}

We want to conclude the examination of extrema for $N_a$ with the following proposition, which shows that $N_a$ shares all its extrema with the Colless index -- not only the maximum $T_n^{cat}$, but also the sets of minimal trees induced by $C$ and $N_a$ coincide. 

\begin{proposition}\label{prop:NaCollessMin}
    Let $n\in\mathbb{N}$ and let $T\in \bt_n$. Then, $T$ minimizes $N_a$ if and only if $T$ minimizes $C$.
\end{proposition}

\begin{proof} First, recall that if $T$ is a Colless minimum, it is also a Sackin minimum (cf. Section~\ref{sec:prelim_Colless}) as well as an $-N_b$ minimum by Theorem~\ref{thm:Dmin}. Therefore, if $T$ is a Colless minimum, both summands in $N_a=S+(-N_b)$ are minimal, which shows that indeed every Colless minimum also minimizes $N_a$.

For the other direction, seeking a contradiction, we assume $T=(T_a,T_b)$ minimizes $N_a$ but is not a Colless minimum. Without loss of generality we may assume $T$ is minimal with this property. Then we distinguish two cases based on $n_a$ and $n_b$ -- the numbers of leaves of $T_a$ and $T_b$ with $n_a\geq n_b$:

\begin{enumerate}
\item First, consider the case $n_a-n_b\leq n_a^{gfb}-n_b^{gfb}$. As $T$ minimizes $N_a$, so do $T_a$ and $T_b$, and by the minimality of $T$, $T_a$ and $T_b$ are Colless minimal. This implies that $(n_a,n_b)$ cannot be a Colless partition (else $T$ would be a Colless minimum), which in particular shows $(n_a,n_b)\neq (n_a^{gfb},n_b^{gfb})$. Thus, it cannot be the gfb partition.

We now argue that it cannot be the echelon partition either. Let $k_n=\lceil\log_2n\rceil$. If $n\geq 3\cdot 2^{k_n-2}$, we have already seen that in this subcase the echelon partition coincides with the gfb-partition, which we have already excluded.

If, however, $n< 3\cdot 2^{k_n-2}$, just as in the second part of the proof of Proposition \ref{prop:echelon} we have $n_a-n_b<n_a^{gfb}-n_b^{gfb}<n_a^{be}-n_b^{be}$, implying that $(n_a,n_b)$ cannot be the echelon partition in this case, either.

As $(n_a,n_b)$ is neither a Colless partition nor the echelon partition, by Theorem~\ref{thm:characterization}, $T$ cannot minimize $-N_b$. Now, let $T'$ be a Colless minimum (which also implies Sackin minimality). Then we know by Theorem~\ref{thm:Dmin}, that $T'$ minimizes $-N_b$. Thus, we get:
\begin{align*}
N_a(T)&=S(T)+(-N_b(T))=S(T')+(-N_b(T))>S(T')+(-N_b(T'))=N_a(T'),
\end{align*}
contradicting the $N_a$ minimality of $T$. 
\item If $n_a-n_b>n_a^{gfb}-n_b^{gfb}$, we know from Section~\ref{sec:prelim_sackin} that $T$ is not Sackin minimal. 

Let $T'$ be a Colless minimum. As above, $T'$ then minimizes $S$ and $-N_b$, too. Thus, we get:
\begin{align*}
N_a(T)&=S(T)+(-N_b(T))>S(T')+(-N_b(T'))=N_a(T'),
\end{align*}
contradicting the $N_a$ minimality of $T$.
\end{enumerate}

As both cases lead to a contradiction, the assumption was wrong. This completes the proof.
\end{proof}

\begin{remark}Before we finish our analysis of $N_a$, we note that both $N_a$ and $N_b$, just like the Sackin index, can be written in terms of a sum of depths of leaves $x \in V^1(T)$ as follows. 

Consider $T$ to be an ordered tree such that for each inner vertex $v$ with children $v_a$ and $v_b$, whenever $n_{v_a}\geq n_{v_b}$, we have that $v_a$ is the \enquote{left} child of $v$ and $v_b$ is the \enquote{right} child of $v$. In this way, each edge in $T$ can be labeled with $L$ or $R$, depending on whether it is leading to the left or the right child of its tail vertex. 

Now, if we consider a leaf $x \in V^1$, this implies that the unique path $P_x$ from the root $\rho$ of $T$ to $x$, which defines the depth $\delta_x$, can be subdivided into $L$-edges and $R$-edges. We denote by $\delta_x^L$ and $\delta_x^R$ the number of $L$- and $R$-edges on this path, respectively.
Then, we get:

$$N_a=\sum\limits_{v\in\mathring{V}(T)}n_{v_a}=\sum\limits_{x\in V^{1}(T)}\delta^L_x,$$ and 

$$N_b=\sum\limits_{v\in\mathring{V}(T)}n_{v_b}=\sum\limits_{x\in V^{1}(T)}\delta^R_x.$$

This is due to the fact that for all vertices $v$ on $P_x$, $x$ is counted in $n_{v_a}$ if and only $x$ is a descendant of the left child of $v$, and it is counted in $n_{v_b}$ if and only if it is a descendant of the right child of $v$. Note that the Sackin index can then be re-written as:

$$S(T)=\sum\limits_{v\in\mathring{V}(T)}n_{v}=\sum\limits_{x\in V^{1}(T)}\delta_x=\sum\limits_{x\in V^{1}(T)}\delta_x^R+\sum\limits_{x\in V^{1}(T)}\delta_x^L.$$
\end{remark}

Finally, we discuss a few more properties of $N_a$. Like the other three proposed indices -- namely $\CS$, $\SC$ and $N_b$ -- $N_a$ can also be computed in linear time.
\begin{proposition}
    For every $T\in\bt_n$, $N_a(T)$ can be computed in time $O(n)$.
\end{proposition}
\begin{proof}
As stated in Sections~\ref{sec:prelim_sackin} and Proposition~\ref{prop:cs_linear_time}, both the Sackin  index and $N_b$ can be computed in time $O(n)$. Since $N_a=S-N_b$ is a linear combination of Sackin and $N_b$, it, too, can be computed in time $O(n)$.
\end{proof}

Next, we state that $N_a$ is local. 

\begin{proposition}\label{prop:Na_local}
   The $N_a$ index is local.
\end{proposition}

We omit the proof as it is completely analogous to that of Proposition \ref{prop:cs_local}, except that $N_a$ now plays the role of $\Delta_{CS}$, $S$ plays the role of $C$ and $N_b$ plays the role of $S$.

\par\vspace{0.3cm}
We finish this subsection by considering recursiveness once more. 

\begin{proposition}\label{prop:CS_brtss2}
   $N_a$ is a binary recursive tree shape statistic.
\end{proposition}
\begin{proof}
Using the recursive property of Sackin and $N_b$, we get
 \begin{align*}
            N_a(T) &=S(T)-N_b(T) \\
            &=S(T_a)-N_b(T_a)+S(T_b)-N_b(T_b) +(n_a+n_b)-\frac{1}{2}((n_a+n_b)-|n_a-n_b|)\\
            &= N_a(T_a)+N_a(T_b)+ \frac{1}{2}((n_a+n_b)+|n_a-n_b|) 
    \end{align*}
    for $T\in\bt_n$ with $n\geq2$ and $N_a(T)=0$ for $T\in\bt_1$.

   Hence, $N_a$ can be written as a binary recursive tree shape statistic of length 2 as follows: 
    \begin{itemize}
        \item $\lambda =(0,1)$ as base case for trees with $n=1$ leaf,
        \item for trees $T=(T_a,T_b)$ we have recursions $r_1(T_a,T_b) = N_a(T_a)+N_a(T_b)+\frac{1}{2}(|n_a-n_b|+(n_a+n_b))$ ($N_a$ value) and $r_2(T_a,T_b)=n_a+n_b$ (leaf numbers).
    \end{itemize}
    Obviously, $\lambda\in\mathbb{R}^2$ and $r_i:\mathbb{R}^2\times\mathbb{R}^2\rightarrow \mathbb{R}$. Furthermore, $r_1,r_2$ are independent of the order of subtrees. 
    \end{proof}

Since the echelon tree plays an important role in our manuscript and as the existing literature on it is very sparse, we will now present several new results on the echelon tree.

    \subsection{ \texorpdfstring{Additional results on echelon trees $T^{be}_n$}{Additional results on echelon trees Tben}}\label{sec:echelon_results} 

    It is the main aim of this subsection to provide an explicit, non-recursive construction for the echelon tree. In order to do so, we need a few more notations. However, before we start, recall that it was already observed in \cite{Currie2024} that the echelon tree is a so-called rooted binary weight tree (which we will define subsequently), but no details were given there.

    Now, let again $k_n=\lceil \log_2(n)\rceil$ and consider the binary expansion of $n$, i.e., $n=\sum\limits_{i=0}^{k_n}\alpha_i\cdot 2^i$, where $\alpha_i$ is the $i+1$-th digit of the binary representation of $n$. Note that $w(n)=\sum\limits_{i=0}^{k_n}\alpha_i$ is the so-called \emph{binary weight} of $n$. Let $f_n(i)$ for $i=1,\ldots,w(n)$ denote the $i^{th}$ power of two in the binary expansion of $n$ when they are sorted in ascending order. For instance, if $n=11=2^0+2^1+2^3$, we have $f_{11}(1)=2^0$, $f_{11}(2)=2^1$ and $f_{11}(3)=2^3$. Note that this implies $n=\sum\limits_{i=1}^{w(n)}f_n(i)$. 
 
    We now construct a tree $T_n^*$ as follows: We start with a caterpillar $T_{w(n)}^{cat}$ on $w(n)$ leaves  -- which we will subsequently refer to as the \enquote{top caterpillar (tree)} --   and replace its leaves by fully balanced trees $T_{\log_2(f_n(i))}^{fb}$ for $i=1,\ldots,w(n)$ in ascending order. More precisely, we replace one of the cherry leaves with $T_j^{fb}$, where $j=\log_2(f_n(1))$ is the smallest index such that $\alpha_j=1$. Then, we take one unvisited leaf of the top caterpillar at a time, starting with the one furthest from the root and ending with the one adjacent to the root, and replace the respective leaf by $T_i^{fb}$, where $i$ is the smallest index with $\alpha_i=1$ which has not been considered yet. Figure~\ref{fig:echelon}  demonstrates this construction.

    We are now in the position to show that $T_n^*$ equals $T_n^{be}$, thus providing an explicit construction for the echelon tree.
    \begin{proposition}\label{prop:explicitechelon} For all $n\in \mathbb{N}_{\geq 1}$ we have $T_n^*=T_n^{be}$. 
    \end{proposition}
    
  \begin{proof} We start by considering $n=2^m$ for some $m\in\mathbb{N}$.
     From Section~\ref{sec:prelim_be}, we know that in this case, the echelon tree coincides with the fully balanced tree of height $m$. Now, consider $T^*_n$. We have $w(n) = 1$ and $f_n(1) = 2^{\log_2(n)}=2^m$. Therefore, the top caterpillar tree used in the construction of $T^*_n$ has to be $T^{cat}_1$ (i.e., the tree consisting of only one vertex and no edges), and the tree that replaces the only leaf in $T_1^{cat}$ is $T^{fb}_m$, implying $T^*_n = T^{fb}_m = T^{be}_n$. This completes the proof for the case in which $n$ is a power of two.

Next, we consider the case in which $n$ is not a power of two and prove the statement by induction on $n$, starting with $n=3$. In the following, let $k_n=\lceil\log_2(n)\rceil$.
    For the base case $n=3=2^1+2^0$ we have $w(3) =2$ with $f_3(1) = 2^0$ and $f_3(2) = 2^1$. Hence, the top caterpillar of $T^*_3$ is $T^{cat}_2$ and the trees that replace its two leaves are $T^{fb}_{\log_2(f_3(1))} = T^{fb}_0$ and  $T^{fb}_{\log_2(f_3(2))} = T^{fb}_1$, resulting in  $T^*_3=T^{cat}_3$. On the other hand, $k_3 = 2$, and by the recursive definition of the echelon tree, $T_3^{be}$ consists of the maximal pending subtrees $T^{fb}_{k_n-1} = T^{fb}_{2-1}=T_1^{fb}$ and $T^{be}_{n-2^{k_n-1}}=T^{be}_{3-2^{2-1}}=T^{be}_{1}$. Since, by definition, $T^{be}_{1}$ is a single vertex, we have $T^{be}_3 = T^{cat}_3 = T^*_3$, completing the base case.

    Before we continue with the inductive step, note that as we consider the case in which $n$ is not a power of two, we know that $w(n)\geq 2$. Moreover, in this case we know that in the binary expansion of $n$, the largest term is $2^{k_n-1}$, i.e., we have $f_n(w(n))=2^{k_n-1}$. We will use this fact now in the inductive step.
    
    So for the inductive step, we now assume that $T^*_k = T^{be}_k$ for all $k<n$ and consider $T^*_n$ with standard decomposition $T^*_n=((T^*_n)_a,(T^*_n)_b)$. As $w(n)\geq 2$, we have $(T^*_n)_a = T^{fb}_j$ with $j= \log_2(f_n(w(n)))=k_n-1$.  Hence, $(T^*_n)_a = T^{fb}_{k_n-1}$. Note that this already implies that $(T_n^*)_a=T_a^{be}$.
    
    It remains to show that $(T^*_n)_b=T_{n-2^{k_n-1}}^{be}$. In order to do so, we first show that $(T^*_n)_b = T^{*}_{n_b}$, where $n_b=n-2^{k_n-1}$. $(T^*_n)_b$ contains $T^{cat}_{w(n)-1}$ as top caterpillar tree with trees $T^{fb}_{\log_2(f_n(1))},\ldots,T^{fb}_{\log_2(f_n(w(n)-1))}$ replacing its leaves in ascending order, i.e. $T^{fb}_{\log_2(f_n(1))},T^{fb}_{\log_2(f_n(2))}$ replace the leaves of the cherry and $T^{fb}_{\log_2(f_n(w(n)-1))}$ replaces the leaf that is adjacent to the root of $(T^*_n)_b$. On the other hand, as $n_b = n-2^{k_n-1}$, we know that the binary expansion of $n_b$ equals the binary expansion of $n$ minus the term $2^{k_n-1}$. Hence, $w(n_b) = w(n)-1$ and $f_{n_b}(i) = f_n(i)$ for all $i=1,\ldots,w(n)-1$. This not only means that the top caterpillar in $T^*_{n_b}$ has to have $w(n)-1$ leaves, but also that the fully balanced trees replacing its leaves have heights $\log(f_n(1)),\ldots,\log(f_n(w(n)-1))$, proving that, in fact, $(T^*_n)_b = T^{*}_{n_b}$.
    
    We now know that $T^*_n = (T^{fb}_{k_n-1},T^*_{n-2^{k_n-1}})$ and we can apply the induction hypothesis to $T^*_{n-2^{k_n-1}}$ to obtain $T^*_n = (T^{fb}_{k_n-1},T^{be}_{n-2^{k_n-1}}) = T^{be}_n$. This completes the proof.
    \end{proof}
    
     Before we continue, we note that this direct, non-recursive construction immediately shows that $T_n^{be}$ is a special case of so-called \emph{rooted binary weight trees}, which were introduced in \cite{Kersting2021}. These trees consist of an arbitrary tree with $w(n)$ many leaves whose leaves are replaced with fully balanced trees of sizes $f_n(i)$ for all $i=1,\ldots,w(n)$ in an arbitrary order.

\begin{figure}[ht]
  \centering
  \includegraphics[scale=.17]{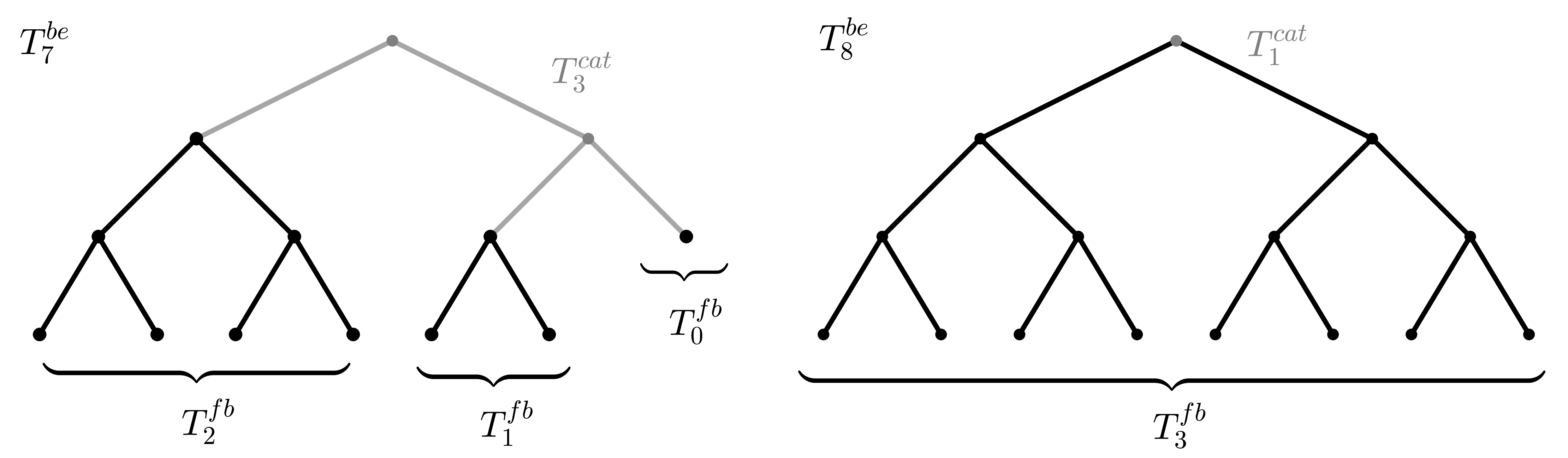}
  \caption{The direct, non-recursive construction of $T_n^{be}$ for $n=7$ and $n=8$. The first step is to calculate the binary weight $w(n)$ of $n$, i.e., to count the numbers of 1's in the binary representation of $n$. In our example we have $n = 7=111_2$ and $n=8 =1000_2$. Hence, $w(7)=3$ and $w(8)=1$. Then, a caterpillar of size $w(n)$ (depicted in gray) is iteratively filled up with fully balanced trees of sizes $2^i$, where the powers of $i$ are those occurring in the binary representation of $n$. The attaching of fully balanced trees starts with the smallest powers of 2 at the cherry and proceeds towards the root. Note that $n=7=2^2+2^1+2^0$ and $n=8=2^3$.}
\label{fig:echelon}
\end{figure}

We conclude this section by providing explicit formulas for the various indices discussed in our manuscript for the echelon tree, which are new to the literature.

\begin{proposition} Let $n\in \mathbb{N}_{\geq 1}$ and let $f_n$ and $w(n)$ be as defined above. Then, we have:

 \begin{align*}S(T_n^{be})&=\sum\limits_{i=1}^{w(n)-1}(f_n(i+1)+f_n(i)\cdot (w(n)-i))+\sum\limits_{i=1}^{w(n)}f_n(i)\cdot\log_2(f_n(i)),\end{align*} 

 \begin{align*}C(T_n^{be})&=\sum\limits_{i=1}^{w(n)-1}(f_n(i+1)-f_n(i)\cdot (w(n)-i)),\end{align*} 

 \begin{align*}\CS(T_n^{be})&=-2\sum\limits_{i=1}^{w(n)-1}f_n(i)\cdot (w(n)-i)-\sum\limits_{i=1}^{w(n)}f_n(i)\cdot\log_2(f_n(i)),\end{align*} 

 \begin{align*}\SC(T_n^{be})&=2\sum\limits_{i=1}^{w(n)-1}f_n(i)\cdot (w(n)-i)+\sum\limits_{i=1}^{w(n)}f_n(i)\cdot\log_2(f_n(i)),\end{align*} 

 \begin{align*}N_a(T_n^{be})&=\sum\limits_{i=1}^{w(n)-1}f_n(i+1)+\frac{1}{2}\sum\limits_{i=1}^{w(n)}f_n(i)\cdot\log_2(f_n(i)),\end{align*} 

 \begin{align*}N_b(T_n^{be})&=\sum\limits_{i=1}^{w(n)-1}f_n(i)\cdot (w(n)-i)+\frac{1}{2}\sum\limits_{i=1}^{w(n)}f_n(i)\cdot\log_2(f_n(i)),\end{align*} 

\end{proposition}

\begin{proof}
Before we start our proof, recall that a rooted binary tree with $w(n)$ leaves has $w(n)-1$ inner vertices. This also  holds, in particular, for the caterpillar top tree used to generate $T_n^{be}$ as depicted in Figure~\ref{fig:echelon}. Since all indices under consideration are a summation over all inner vertices of the tree, we can subdivide these sums into two sums: the sum over the inner vertices of the caterpillar top tree and the sum over the inner vertices of the pending fully balanced subtrees (including the leaves of the caterpillar top tree) of $T_n^{be}$. The sums over the inner vertices of the fully balanced subtrees then gives the respective index for these pending subtrees, for which the values are already known. This approach will be used subsequently.

\begin{itemize}
\item We begin with the Sackin index. Each of the fully balanced subtrees $T_i$ with $f_n(i)=2^i$ leaves have Sackin index $S(T_i)=2^i \cdot \log_2(i)$ (as we know that a fully balanced tree with $n=2^{k_n}$ many leaves has Sackin index $2^{k_n}\cdot k_n$, cf. Section~\ref{sec:prelim_sackin}). Together with the fact that we have $w(n)$ many such subtrees in our tree, this explains the second summand. 

Now, for the first summand, we consider all $w(n)-1$ inner vertices $v_1,\ldots,v_{w(n)-1}$ of the top caterpillar. Let $v_1$ be the parent vertex of the cherry and $v_i$ be adjacent to $v_{i-1}$ for $i=2,\ldots, w(n)-1$ such that $v_{w(n)-1}$ is the root of the top caterpillar. By the construction of $T^{be}_n$, we know that the larger pending subtree of vertex $v_i$ is the fully balanced tree with $f_n(i+1)$ leaves. The smaller pending subtree of $v_i$ contains $T^{fb}_{\log_2(f_n(1))},\ldots,T^{fb}_{\log_2(f_n(i))}$, implying that
\begin{align*}
    \sum\limits_{v\in\Vint(T^{cat}_{w(n)})} n_v  &= \sum\limits_{i=1}^{w(n)-1} \underbrace{\left(f_n(i+1)+ \sum\limits_{j=1}^i f_n(j)\right)}_{n_{v_i}}\\
    &=\left(\sum\limits_{i=1}^{w(n)-1} f_n(i+1)\right)+\left(\sum\limits_{i=1}^{w(n)-1} \sum\limits_{j=1}^i f_n(j)\right)\\
    &=\sum\limits_{i=1}^{w(n)-1} \left(f_n(i+1)+f_n(i)\cdot (w(n)-i)\right).
\end{align*}

This explains the first summand of the term and thus completes the proof for the Sackin index.

\item For the Colless index, we proceed as for the Sackin index. We note, however, that each fully balanced subtree contributes a Colless index of 0. Hence, the respective term that for the Sackin index led to the second sum, is 0 in the Colless case and thus does not need to be considered. 

Therefore, we only have to consider the summation over all inner vertices of the top caterpillar tree of $T_n^{be}$, but instead of adding the sizes of both pending subtrees for each of the vertices, we subtract the size of the smaller subtree from the larger one. This explains the formula and completes the proof for the Colless index.

\item The formulas of $\Delta_{CS}=C-S$ and $\Delta_{SC}=S-C$ now simply follow from the definition of these indices using the already derived formulas for $S$ and $C$.
\item The formula for $N_b$ simply follows from $\Delta_{CS}$ by using Lemma~\ref{lem:D=sumnb}, which shows that $N_b=-\frac{1}{2}\Delta_{CS}$.
\item Finally, as we know that $N_a=S-N_b$, the formula for $N_a$ is a direct conclusion of the already derived formulas for $S$ and $N_b$. This completes the proof.
\end{itemize}
\end{proof}

\color{black}

\begin{remark} Note that the integer sequences 0,2,5,8,13,16,20,24,33,36,40 and 0,0,1,0,3,2,2,0,7,6,6,4,6,4,3,0 of the values of $S(T_n^{be})$ and $C(T_n^{be})$, respectively, for $n\geq 1$ had not been contained in the Online Encyclopedia of Integer sequences \cite{oeis} before our study was conducted. This indicates that they might not have occurred in other contexts yet. They have been added to the OEIS in the course of this manuscript as sequences A396769 (Sackin) and A397672 (Colless). Interestingly, the sequence 0,1,3,4,8,9,11,12,20,21,23,24 of $N_a(T_n^{be})$ was, at least without the first term of 0 for the trivial tree consisting of only one vertex, already contained in the OEIS \cite[Sequence A006520]{oeis}.  Also, the sequence 1,2,4,5,7,9,12,13,15,17,20,22,25,28,32,33,35,37,40,42,45 of $N_b(T_n^{be})$ was already contained in the OEIS \cite[Sequence A000788]{oeis}, albeit in a shifted version, namely  $N_b(T_n^{be}) = A000788(n-1)$. Due to the close relationship of $\CS$ and $\SC$ to $N_b$, we also have $\CS(T^{be}_n) = -2\cdot A000788(n-1) $ with sequence terms 0,-2,-4,-8,-10,-14,-18,-24,-26,-30 and $\SC(T^{be}_n)= 2\cdot A000788(n-1)$ with sequence terms 0,2,4,8,10,14,18,24,26,30. These identities can easily be verified  by showing the equivalence of the recursive formulas for $\SC(T^{be}_n)$ and $2\cdot A000788(n-1)$. 
The fact that these sequences have already occurred in other contexts nicely links our findings on the echelon tree and the new indices $N_a$ and $N_b$ to other combinatorial research areas.
\end{remark}
 
\section{Discussion} 
In our manuscript, we have introduced two new imbalance indices, namely $N_a$ and $\Delta_{CS}$ as well as two new balance indices, namely $N_b$ and $\Delta_{SC}$, which are all closely related. We have shown that the established Sackin and Colless indices can be regarded as mere compound indices of two of these indices, namely $N_a$ and $N_b$, as we have $S=N_a+N_b$ and $C=N_a-N_b$. As $N_a$ and $N_b$ take less information of the tree into account than $S$ and $C$, we regard them as more basic than $S$ and $C$ themselves. Interestingly, the Colless index and the $N_a$ index share the same extremal sets, whereas $N_b$ also considers, among others, the echelon tree as maximally balanced, which neither $C$ nor $N_a$ do. Thus, $N_a$ seems to be the dominating term in $C=N_a-N_b$, as concerning the extremal sets, $N_b$ does not change the sets induced by $N_a$. However, we do note that there are pairs of trees $(T_1,T_2)$ such that $N_a(T_1)>N_a(T_2)$ and $C(T_1)<C(T_2)$ (results not shown), implying that the rankings induced by $N_a$ and $C$ are \emph{not} identical, even though the indices agree on their extremal sets. 

Another interesting aspect is that as $N_a$ is an imbalance index and $N_b$ is a balance index, the fact that $C$ is an imbalance index, too, follows trivially, whereas the fact that the Sackin index is an imbalance index is less obvious, because, by $S=N_a+N_b$, it is a sum of an imbalance and a balance index. However, the fact that it is an imbalance index just like $N_a$ seems to indicate that here, too, $N_a$ is the dominating term. But again, we note that there are pairs of trees $(T_1,T_2)$ with $N_a(T_1)>N_a(T_2)$ but $S(T_1)<S(T_2)$ (results not shown). Thus, again, the two indices' rankings are not identical, showing that $N_a$ is a relevant and interesting index in its own right.

Note that $N_b$'s maximum set contains, among others, echelon trees and trees with echelon subtrees, which also links $N_b$ to the stairs2 balance index from the literature \cite{Colijn2014, Fischer2023,Currie2024}, which is known to have the echelon tree as its only maximum and which is defined as follows: 
    \[st2(T) =\frac{1}{n-1}\sum\limits_{v\in\Vint(T)}\frac{n_{v_b}}{n_{v_a}},\] where $v_a$ and $v_b$ are again the children of $v$, and  $n_{v_a}\geq n_{v_b}$.
Figure~\ref{fig:indexcomparison} shows that $\Delta_{CS}$ (and thus by $\Delta_{CS}=-\Delta_{SC}=-2N_b$, also $\SC$ and $N_b$) are generally less resolved than $S$, $C$, or $st2$, but this is not true for $N_a$. In fact, $N_a$ is sometimes even more resolved than, say, the Sackin index, cf. Figure~\ref{fig:indexcomparison}. However, it can also be seen that when for instance $N_b(T_1)=N_b(T_2)$, we often have that the rankings of $S$ and $C$ do not coincide with the one of $st2$. While this is clearly not always the case, equal rankings by $N_b$, for instance, might be used as a first indicator that there is some ambivalence with the considered trees and thus the index to be used should be carefully selected for the purpose at hand. 

We conclude by noting that concerning $\Delta_{CS}$, whenever the rankings of two trees $T_1$ and $T_2$ are different between $S$ and $C$, i.e., if we have $S(T_1)>S(T_2)$ but $C(T_1)<C(T_2)$, then $\Delta_{CS}(T_1)=C(T_1)-S(T_1)<C(T_2)-S(T_2)=\Delta_{CS}(T_2)$. Hence, in such cases, $\Delta_{CS}$ (and thus $N_b$) will agree with the Colless index and not with the Sackin index concerning which tree is more balanced. Thus, it seems that for the difference $\Delta_{CS}=C-S$, indeed the Colless index is more dominant than the Sackin index (which can also be seen by Theorem~\ref{thm:characterization}, as it was shown there that a Sackin minimum which is not Colless minimal cannot minimize $\Delta_{CS}$ or maximize $N_b$, whereas all Colless minima are guaranteed to do so).
An interesting question arising from our manuscript is that of finding a direct, non-recursive formula for $a(n)$, the number of $\Delta_{CS}$ minima (and $N_b$ and $\Delta_{SC}$ maxima). We leave this problem open for future research.
Last but not least, we also added to the literature by investigating the echelon tree more in-depth, which is known particularly from the context of the stairs2 and $B_1$ indices. We hope that the explicit definition presented in this manuscript will facilitate future research concerning this tree.

\begin{figure}[ht]
  \centering
  \includegraphics[scale=.55]{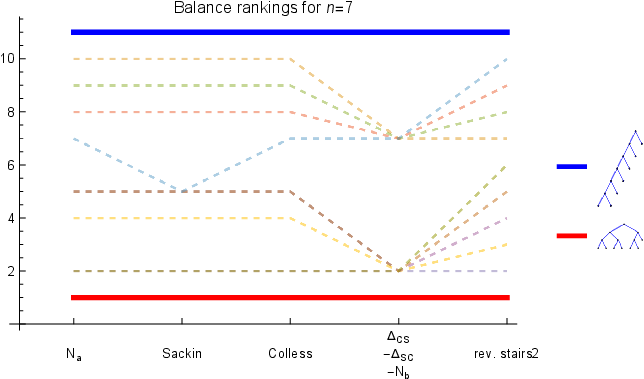}\hspace{0.03\textwidth}
  \includegraphics[scale=.55]{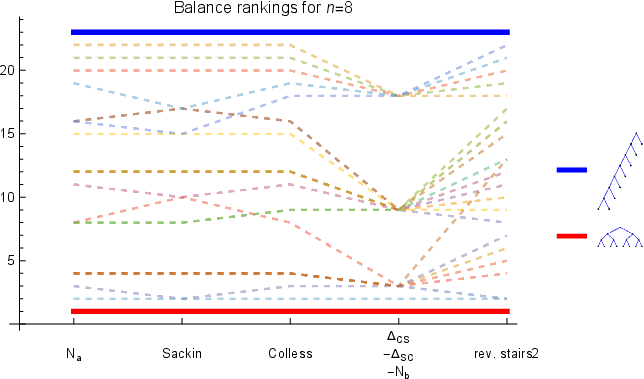}\hspace{0.03\textwidth}
  \includegraphics[scale=.55]{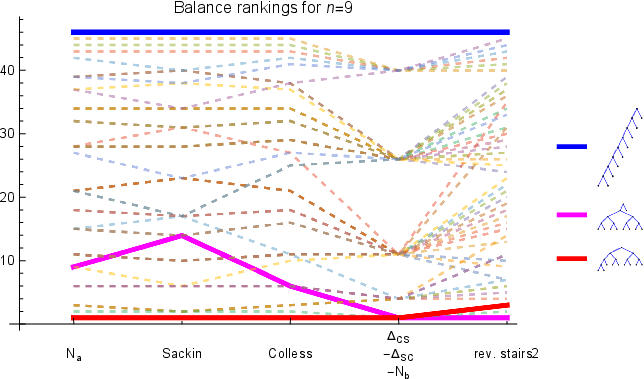}
  \caption{Ranking comparisons of the indices $N_a$, $S$, $C$, $\Delta_{CS}$ (whose induced rankings are the same as those of $-\Delta_{SC}$ and $-N_b$) and stairs2. Note that the depicted plots were generated with the computer algebra system Mathematica \cite{mathematica}. The $y$-axis contains all trees ordered according to the so-called Furnas rank \cite{Furnas1984,Fischer2023,Kirkpatrick1993}, starting with the maximally balanced tree and proceeding to the caterpillar. The blue line corresponds to the caterpillar, the red line corresponds to the maximally balanced tree, which in case $n=7,8,9$ also coincides with the gfb tree. For $n=7,8$ the echelon tree also coincides with the maximally balanced tree, whereas it differs from it -- and is thus depicted with its own magenta line -- for $n=9$. Note that the stairs2 ranking was inverted (by using $-st2$ instead of $st2$) in order to make it an imbalance index as the others in the graphic for comparability. It can clearly be seen that $\Delta_{CS}=C-S$ is the least resolved index. However, it can also be seen that $\CS$ often ranks trees equally which are ranked differently between Colless on the one hand and stairs2 on the other hand -- see, for instance, the light-blue line in the figure for $n=7$ as an extreme example.  }
\label{fig:indexcomparison}
\end{figure}

\backmatter

\section*{Acknowledgment} 
MF and LK were supported by the project ArtIGROW, which is a part of the WIR!-Alliance \enquote{ArtIFARM – Artificial Intelligence in Farming}, and gratefully acknowledge the Bundesministerium für Forschung, Technologie und Raumfahrt (Federal Ministry of Research, Technology and Space, FKZ: 03WIR4805) for financial support.
The authors also wish to thank Kristina Wicke and Sophie Kersting for helpful discussions and various insights. Some parts of the material presented in this manuscript are based upon work supported by the National Science Foundation under Grant No.~DMS-1929284 while MF was in residence at the Institute for Computational and Experimental Research in Mathematics in Providence, RI, during the Theory, Methods, and Applications of Quantitative Phylogenomics semester program. 

\section*{Conflict of interest} The authors herewith certify that they have no affiliations with or involvement in any organization or entity with any financial (such as honoraria; educational grants; participation in speakers’ bureaus; membership, employment, consultancies, stock ownership, or other equity interest; and expert testimony or patent-licensing arrangements) or non-financial (such as personal or professional relationships, affiliations, knowledge or beliefs) interest in the subject matter discussed in this manuscript.

\section*{Data availability statement} 
Data sharing is not applicable to this article as no new data were created or analyzed in this study.

\section*{Authors' contributions} MF formulated the main research question and developed the theoretical framework. All authors contributed equally to the  development and writing of the manuscript. All authors read and approved the final version.
\bibliography{literature}
\end{document}